\documentclass[preprint,12pt]{elsarticle}

\usepackage{amssymb}
\usepackage{amsmath}
\usepackage{amsfonts}
\usepackage{mathtools}
\usepackage{bm} 
\usepackage{url}

\usepackage{stmaryrd} 
\usepackage{amsthm}   

\usepackage{enumitem} 
\usepackage[utf8]{inputenc}
\usepackage[T1]{fontenc}

\theoremstyle{plain}
\newtheorem{theorem}{Theorem}[section]
\newtheorem{lemma}[theorem]{Lemma}
\newtheorem{proposition}[theorem]{Proposition}
\newtheorem{corollary}[theorem]{Corollary}

\theoremstyle{definition}
\newtheorem{definition}[theorem]{Definition}

\newtheorem{remark}[theorem]{Remark}

\begin{document}
	
	\begin{frontmatter}
		
		\title{A Dual-Threshold Probabilistic Knowing Value Logic}
		
		\author[label1]{Shanxia Wang}
		\ead{wangshanxia@htu.edu.cn}
		
\affiliation[label1]{
	organization={School of Computer and Information Engineering (School of Artificial Intelligence)},
	addressline={Henan Normal University}, 
	city={Xinxiang},
	postcode={453007}, 
	state={Henan},
	country={China}
}
		
		\begin{abstract}
			A central problem in knowledge representation for uncertain multi-agent settings is how to accommodate, within a single formalism, both probabilistic-threshold attitudes toward propositions and high-confidence attitudes toward term values. Existing probabilistic epistemic logics mainly address graded attitudes toward propositional events, whereas classical knowing-value logics focus on value determination under relational semantics. This paper bridges the gap by introducing a \textbf{dual-threshold probabilistic knowing value logic}. The framework is especially relevant in privacy-sensitive settings, where an attacker may assign high posterior probability to one candidate value of a sensitive attribute without that value being guaranteed to be the actual one.
			
			The key idea is to separate the threshold domains of propositional and value-oriented operators. The operator $K_i^\theta$ retains the full rational threshold range, whereas the knowing-value operator $Kv_i^\eta(t)$ is restricted to the high-threshold interval $(\frac{1}{2},1]$. This restriction is structural. Once $\eta>\frac{1}{2}$, two distinct values cannot both satisfy the threshold, so the uniqueness condition for knowing-value becomes automatic. Over probabilistic models with countably additive measures, $Kv_i^\eta(t)$ is interpreted as \textbf{non-factive high-confidence value locking}.
			
			On the proof-theoretic side, the paper establishes sound axiomatic systems for the framework and develops a two-layer construction based on \textbf{type-space distributions and assignment-configuration mappings}. This construction resolves the joint realization problem generated by probabilistic mass allocation and value-sensitive constraints in multi-agent, multi-term settings. As a result, the paper proves a \textbf{structured weak-completeness} theorem for the high-threshold fragment. The framework thus provides a unified formal treatment of probabilistic reasoning about propositions and high-confidence knowing-value reasoning about term values in multi-agent settings.
		\end{abstract}


	\end{frontmatter}

\section{Introduction}
\label{sec:intro}

A central problem in artificial intelligence, knowledge representation, and multi-agent systems is how to represent agents' epistemic states in a way that is both expressive and formally well behaved. Classical epistemic logic provides a natural framework for factive \emph{know-that} attitudes. Its Kripke-style relational semantics, however, is better suited to binary epistemic distinctions than to graded confidence, subjective uncertainty, or probabilistic judgment \cite{Fagin1995RAK}. Probabilistic epistemic logic addresses this limitation by enriching possible-world semantics with probability measures. In this way, it yields a finer-grained account of belief, confidence, and uncertainty about events, and has become a standard framework for uncertain reasoning, multi-agent decision making, and adversarial modeling \cite{Fagin1990Prob,Halpern1993Know,Fagin1994Reason,Halpern2003Uncert,Liu2023Progression}. In many AI applications, this additional granularity is indispensable.

Yet in many epistemic tasks, the issue is not merely whether a proposition is likely to be true, but whether an agent can identify the \emph{value} of an object, variable, or parameter. In privacy analysis, for example, one may ask whether an attacker can lock onto one specific candidate value of a sensitive attribute with sufficiently high confidence. More generally, similar questions arise whenever an agent attempts to identify the value of a hidden but decision-relevant parameter. Such attitudes are inherently value-oriented, and they are not straightforwardly reducible to ordinary propositional knowledge. Under relational semantics, their defining feature is that the agent can exclude competing candidates and thereby identify the value of a term. Knowing-value logic has studied this phenomenon in depth \cite{Baltag2016Know,Gu2016KV,vanEijck2017Inspect,Wang2022TermModal}. What remains missing is a natural integration of knowing-value reasoning with probabilistic-threshold cognition.

\paragraph{Illustrative privacy scenario.}
Consider a privacy attacker who tries to infer the value of a sensitive attribute $t$, such as a disease category or an income bracket. Suppose that the attacker assigns posterior probabilities $0.62$, $0.23$, and $0.15$ to the candidate values of $t$. It is then natural to say that the attacker has locked onto one specific candidate value with high confidence, even though that candidate value need not be the actual one at the current world. This illustrates the intended non-factive reading of probabilistic knowing-value. It also shows why value-oriented cognition cannot be captured by propositional threshold acceptance alone: $K_i^\theta\varphi$ may express that a proposition is sufficiently probable, but it does not express that one specific candidate value of a sensitive attribute is uniquely supported above the threshold. By contrast, $Kv_i^\eta(t)$ is designed to express exactly such unique high-confidence support. The same example also motivates the dual-threshold design. If the posterior probabilities were instead $0.42$, $0.37$, and $0.21$, lowering the threshold could still support acceptance of various propositions, but it would no longer support high-confidence locking onto one specific sensitive-attribute value. This is why the knowing-value operator must be restricted to the high-threshold interval $(\frac{1}{2},1]$.

This integration faces two obstacles. The first is semantic. Classical knowing-value semantics is defined in terms of universal agreement across epistemically possible worlds. In a probabilistic setting, that idea must be reconstructed in terms of sufficiently concentrated probability mass. The second obstacle concerns thresholds. For propositional threshold operators, lowering the threshold typically enlarges the set of propositions an agent is prepared to accept. For knowing-value operators, by contrast, lowering the threshold may allow several competing values to cross the threshold simultaneously. This undermines the uniqueness condition built into knowing-value semantics. A uniform treatment of thresholds is therefore unstable at the metatheoretic level.

This paper addresses this problem by proposing a \textbf{dual-threshold probabilistic knowing value logic}. The guiding idea is simple but structurally important: propositional cognition and value cognition are governed by different threshold domains. The operator $K_i^\theta$ retains the full range $\theta\in[0,1]$, so that probabilistic reasoning about propositions preserves its usual expressive flexibility. By contrast, the knowing-value operator $Kv_i^\eta(t)$ is restricted to the high-threshold interval $(\frac{1}{2},1]$. This restriction is structural. Once $\eta>\frac{1}{2}$, two distinct values cannot both reach the threshold, so the uniqueness requirement for knowing-value becomes automatic. In this way, the instability caused by multiple competing low-threshold candidates is eliminated, yielding a viable basis for axiomatization.

Semantically, we interpret the logic on probabilistic models with countably additive probability measures. The operator $Kv_i^\eta(t)$ is read as a form of \textbf{non-factive high-confidence value locking}. An agent may uniquely assign sufficiently high probability to one candidate value without that value coinciding with the actual value at the current world. This gives the propositional and value-oriented operators a \textbf{coherent} non-factive epistemic profile. On the proof-theoretic side, the main technical challenge is to handle value conflicts and probability-mass allocation in multi-agent, multi-term settings. To address this challenge, we introduce a reconstruction method based on \textbf{assignment-configuration mappings}. Rather than fixing a privileged value for each agent in advance, we consider probability distributions over possible term--value assignments. Together with a joint value-fiber consistency condition \textbf{(Fib)}, this yields a two-layer construction based on \textbf{type-space distributions and assignment-configuration mappings}. The construction supports a proof of \textbf{structured weak completeness} for the high-threshold fragment.

The main contributions of the paper are as follows:
\begin{enumerate}
	\item We propose \textbf{a dual-threshold probabilistic knowing value logic}. The dual-threshold design reconciles expressive adequacy for probabilistic propositional cognition with structural stability for knowing-value semantics.
	
	\item We introduce \textbf{a non-factive probabilistic semantics for knowing value}. This semantics yields a \textbf{coherent} epistemic interpretation for propositional and value-oriented operators, while making uniqueness automatic above $\frac{1}{2}$.
	
	\item We establish corresponding \textbf{axiomatic systems and soundness results}. These systems isolate the core valid principles of the framework and show that the proof theory is not an ad hoc juxtaposition of probabilistic and value-oriented components.
	
	\item We develop \textbf{a structured weak-completeness method based on type-space distributions and assignment-configuration mappings}. The truth lemma is established by a single induction on modal depth, which ensures that the probability computation at each level relies only on the semantic identification at strictly lower levels, thereby resolving the mutual dependence between measure realization and truth evaluation. This method shows that probabilistic mass allocation and multiple value constraints can still be handled within a single metatheoretic framework.
\end{enumerate}

Overall, the paper advances a single methodological claim: probabilistic reasoning about propositions and high-confidence reasoning about values can be treated within one multi-agent formalism, provided that the asymmetry between their threshold behaviors is built into the language and reflected in the semantics and completeness construction. The rest of the paper is organized as follows. Section~\ref{sec:related} reviews the most relevant related work. Section~\ref{sec:lang} introduces the language, models, and satisfaction relation. Section~\ref{sec:properties} establishes the main semantic properties. Section~\ref{sec:axioms} presents the axiomatic systems, and Section~\ref{sec:soundness} proves their soundness. Section~\ref{sec:completeness} develops the type-space and assignment-configuration construction and proves structured weak completeness for the high-threshold fragment. Section~\ref{sec:conclusion} concludes.


\section{Related Work}
\label{sec:related}

This paper lies at the intersection of probabilistic epistemic logic,
semantic and completeness methods for probabilistic belief models, and
knowing-value logic together with more general value-oriented epistemic
logics. Each of these areas is already well developed. What is still
missing, however, is a unified framework that treats graded probabilistic
attitudes toward propositions and high-confidence attitudes toward term
values within the same multi-agent language, under a common semantic and
proof-theoretic treatment. The present paper is intended to fill that gap.


A first line of work concerns probability, knowledge, and uncertainty.
Against the broader background of epistemic logic in multi-agent systems
\cite{Fagin1995RAK}, Fagin, Halpern, and Megiddo developed the basic
logical treatment of probabilistic constraints over possible-world models
\cite{Fagin1990Prob}. Halpern and Tuttle then studied the interaction of
knowledge and probability in multi-agent and adversarial settings
\cite{Halpern1993Know}, while Fagin and Halpern further clarified the
logic of combined epistemic and probabilistic reasoning
\cite{Fagin1994Reason}. Halpern also showed that knowledge should not be
identified with probability-$1$ belief \cite{Halpern1991Certain}, and
later provided a broad synthesis of reasoning under uncertainty
\cite{Halpern2003Uncert}. More recent work has extended this tradition in
several directions, including countably additive semantics,
neighborhood-style semantics, strong completeness, decidability,
likelihood and belief, conditional probability, and belief progression
\cite{Doder2024Countable,PanGuo2024,Ognjanovic2024Strong,Ognjanovic2024Dec,Delgrande2022ELLB,Dautovic2023Cond,Liu2023Progression}.
Computational developments have also appeared, for example in model
checking for probabilistic epistemic logics over probabilistic multi-agent
systems \cite{Fu2018MCPETL}. This literature provides a mature account of
graded attitudes toward \emph{propositions}. What it does not usually
provide is a semantics for high-confidence identification of a
\emph{specific value} of a term. Our framework preserves the
probabilistic-threshold treatment of propositions, but extends the
analysis from events to values.


A second line of work concerns the semantic machinery used to represent
probabilistic epistemic states and to prove completeness. Type-space
semantics has been especially important in this context. Heifetz and
Mongin studied probability logic on Harsanyi type spaces
\cite{Heifetz2001Type}. Meier developed an infinitary probability logic
and established strong completeness results \cite{Meier2012Infin}. Zhou
further investigated probability logic for Harsanyi type spaces
\cite{Zhou2014Type}. The general lesson of this line of work is that, once
probabilistic constraints enter the object language, plain relational
semantics is often too coarse for completeness-theoretic purposes, and
richer local probabilistic structure becomes necessary.

A common methodological motif in these completeness arguments is the
\emph{realizability of local probabilistic constraints}. Fagin and
Halpern's completeness proof for the logic of combined epistemic and
probabilistic reasoning already follows this pattern: finite conjunctions
of linear threshold conditions derived from a consistent formula are shown
to be simultaneously satisfiable by a local probability measure at each
world \cite{Fagin1994Reason}. The type-space constructions of Heifetz and
Mongin \cite{Heifetz2001Type} and Zhou \cite{Zhou2014Type} extend this
motif by showing that locally coherent probability assignments can be
assembled into a global model. This perspective directly informs the
present paper. Our completeness construction also relies on local
probabilistic realizability, but it must additionally track term equality
and support value-sensitive operators. For that reason, the construction
adds a second layer on top of local type-space distributions, namely an
\emph{assignment-configuration layer} organized by value-fiber constraints.
In methodological terms, the paper therefore extends the local-constraint
realizability paradigm of probabilistic epistemic completeness to a
setting in which probabilistic mass-allocation constraints and
value-theoretic uniqueness constraints must be realized jointly.


A third line is the literature on knowing-value and related
value-oriented epistemic logics. Baltag's proposal that to know is to
know the value of a variable established knowing-value as an epistemic
notion in its own right \cite{Baltag2016Know}. Gu and Wang showed that
knowing-value logic admits a normal modal treatment \cite{Gu2016KV}, and
van Eijck, Gattinger, and Wang connected knowing-value to public
inspection and epistemic update \cite{vanEijck2017Inspect}. Later work
broadened the value-oriented perspective in several directions. Functional
dependence logics study structural relations between values
\cite{Baltag2021LFD}. Dynamic data-informed knowledge links value
information to update \cite{Deuser2024Data}. Richer object-level settings
have been developed through epistemic term-modal logics with assignment,
predicate-value operators, temporal knowing-value formalisms, and dynamic
term-modal approaches to epistemic planning
\cite{Wang2022TermModal,Hong2023Predicate,Lin2021LTLKV,Liberman2020Planning}.
Taken together, these studies show that epistemic reasoning about values
is now a substantial topic in its own right.

Still, most knowing-value frameworks remain relational in character: an
agent knows the value when all relevant accessible worlds agree on it.
The present paper departs from this tradition in two respects. First, we
reconstruct knowing-value in a probabilistic setting as
\emph{high-confidence value locking}: the key issue is no longer universal
agreement across accessible worlds, but whether there is a uniquely
supported value above a suitably high threshold. Second, the resulting
semantics is \textbf{non-factive}: the value that the agent locks onto
need not coincide with the actual value at the current world. This
contrasts with the standard relational reading, where knowing-value
entails that the identified value is in fact the correct one.


A fourth line concerns AI-oriented extensions of epistemic and
probabilistic reasoning, especially in dynamics, planning, revision, and
verification. Probabilistic dynamic epistemic logic and probabilistic
conformant planning already show how epistemic probability can interact
with update and action \cite{Sack2009PDEL,Li2019Prob}. On the
non-probabilistic but closely related side, recent work in epistemic
planning has studied decidability, complexity, tractable fragments,
higher-order belief, justified perspectives, observation, revision, and
meta-level planning
\cite{Bolander2020DEL,Cooper2021Lightweight,Wan2021MAEP,Muise2022NestedBelief,Hu2023JustifiedPerspectives,Engesser2024EDP,LiWang2024PlanAboutPlanning,Belle2023SpecialIssue}.
The 2025 literature pushes this line further. Probabilistic knowing-how
has been studied as a graded epistemic attitude \cite{Castro2025ProbKH}.
Probabilistic belief revision has been investigated as a formal update
problem \cite{Delgrande2025ProbRevision}. Belief-based programs have also
been connected to verification under noisy actions and sensing
\cite{LiuLakemeyer2025BeliefPrograms}. These developments matter here for
two reasons. First, they show that AI increasingly requires epistemic
formalisms that are graded, revisable, and computationally meaningful.
Second, they show that value-sensitive information is not peripheral: once
planning, revision, and verification are carried out in object-rich
multi-agent settings, reasoning about specific values becomes a natural
part of the epistemic state.


Against this background, the gap addressed in this paper can be stated
directly. Existing probabilistic epistemic logics provide tools for graded
reasoning about propositions, but not for probabilistic knowing-value over
terms. Existing knowing-value logics provide tools for reasoning about
values, but typically under relational agreement semantics rather than
probabilistic confidence semantics. Existing semantic and completeness
techniques for probabilistic logic provide important guidance, but they do
not by themselves resolve the interaction between probabilistic thresholds
and value-sensitive uniqueness conditions. The present paper addresses this
gap by combining these strands within a single framework. It introduces a
dual-threshold language that separates propositional probability thresholds
from high-threshold knowing-value thresholds; it gives knowing-value a
non-factive probabilistic semantics in terms of unique high-confidence
support; and it develops a structured weak-completeness argument based on
type-space distributions and assignment-configuration mappings. The
contribution is therefore not merely to juxtapose probabilistic epistemic
logic and knowing-value logic, but to \textbf{provide a unified
	representational and proof-theoretic framework} that bridges them.


\section{Formal Language and Semantics}
\label{sec:lang}

\subsection{Syntax}
\label{subsec:syntax}

We consider a multi-agent modal language that simultaneously captures probabilistic-threshold reasoning about propositions and high-confidence knowing-value reasoning. Unlike approaches based on a uniform threshold regime, our framework adopts a \textbf{dual-threshold architecture} already at the level of the object language: the propositional operator $K_i^\theta$ allows any rational threshold in the interval $[0,1]$, whereas the knowing-value operator $Kv_i^\eta(t)$ is restricted from the outset to the high-threshold interval $(\frac{1}{2},1]$. This distinction is motivated by structural rather than merely notational considerations. In the high-threshold range, the uniqueness condition in the semantics of knowing-value becomes automatic, which substantially simplifies the metatheoretic analysis developed later. At the same time, the propositional probabilistic operator retains general thresholds in order to preserve the expressive power of threshold-based probability logic. In this respect, the language continues the explicit treatment of object values familiar from knowing-value and value-oriented epistemic logics, while remaining aligned with the standard treatment of probabilistic threshold operators in probabilistic epistemic logic \cite{Baltag2016Know,Gu2016KV,Wang2022TermModal,Fagin1990Prob}.

Let $\mathsf{Prop}=\{p,q,r,\dots\}$ be a countable set of propositional variables, let $\mathsf{Term}=\{t,s,u,\dots\}$ be a countable set of atomic term symbols, and let $\mathcal{A}=\{1,\dots,n\}$ be a finite set of agents. We define two threshold sets:
\[
\Theta_K=[0,1]\cap\mathbb{Q},
\qquad
\Theta_V^{+}=\Bigl(\frac{1}{2},1\Bigr]\cap\mathbb{Q}.
\]

\begin{definition}[Language $\mathcal{L}_{\mathrm{PTMLKv}}$]
	The set of formulas of $\mathcal{L}_{\mathrm{PTMLKv}}$ is generated by the following grammar:
	\[
	\varphi,\psi ::= p \mid t=s \mid \neg \varphi \mid (\varphi \to \psi)
	\mid K_i^\theta \varphi \mid Kv_i^\eta(t),
	\]
	where $p\in\mathsf{Prop}$, $t,s\in\mathsf{Term}$, $i\in\mathcal{A}$, $\theta\in\Theta_K$, and $\eta\in\Theta_V^{+}$.
\end{definition}

Here, $t=s$ states that the terms $t$ and $s$ have the same value; $K_i^\theta\varphi$ says that agent $i$ assigns probability at least $\theta$ to the proposition $\varphi$; and $Kv_i^\eta(t)$ says that agent $i$ knows the value of the term $t$ with confidence at least the threshold $\eta$.

We use conjunction $\varphi\wedge\psi$, disjunction $\varphi\vee\psi$, biconditional $\varphi\leftrightarrow\psi$, and the constants $\top,\bot$ as standard abbreviations.

\begin{remark}
	In the dual-threshold setting adopted here, propositional operators and knowing-value operators belong to a common language, but their threshold domains are distinct. This distinction is not introduced merely for notational convenience; it is required for the metatheoretic analysis developed later. For $K_i^\theta$, retaining general thresholds makes it possible to express probabilistic monotonicity, complementarity, and additive behavior supported by probability measures. For $Kv_i^\eta(t)$, by contrast, restricting thresholds directly to the high-threshold interval removes from the outset the degeneracy and non-monotonic complications that arise at lower thresholds. More broadly, the present setting combines the expressive strengths of threshold operators in probabilistic logic with the term-centered character of value-oriented epistemic logics \cite{Fagin1990Prob,Halpern2003Uncert,Baltag2016Know,Wang2022TermModal}.
\end{remark}

\subsection{Models}
\label{subsec:models}

\begin{definition}[Multi-agent probabilistic knowing-value model]
	An $\mathcal{L}_{\mathrm{PTMLKv}}$-model is a tuple
	\[
	M=(W,D,\{P_i\}_{i\in\mathcal{A}},V,\mathsf{val}),
	\]
	where:
	\begin{enumerate}
		\item $W\neq\emptyset$ is a nonempty set of possible worlds;
		\item $D\neq\emptyset$ is a nonempty domain of values;
		\item for each agent $i\in\mathcal{A}$ and each world $w\in W$, 
		$P_i(w):\mathcal{P}(W)\to[0,1]$ is a countably additive probability measure on the powerset of $W$, that is:
		\begin{itemize}
			\item $P_i(w)(\emptyset)=0$ and $P_i(w)(W)=1$;
			\item for every countable pairwise disjoint family $\{X_k\}_{k\in\mathbb{N}}$,
			\[
			P_i(w)\Bigl(\bigcup_{k\in\mathbb{N}} X_k\Bigr)
			=
			\sum_{k\in\mathbb{N}} P_i(w)(X_k).
			\]
		\end{itemize}
		\item $V:W\times\mathsf{Prop}\to\{0,1\}$ is a propositional valuation function;
		\item $\mathsf{val}:W\times\mathsf{Term}\to D$ is a term-value function, where $\mathsf{val}(w,t)$ denotes the value of the atomic term $t$ at world $w$.
	\end{enumerate}
	A pair $(M,w)$ with $w\in W$ is called a pointed model.
\end{definition}

This choice of model preserves the basic world--event--measure organization familiar from probabilistic epistemic logic, while explicitly adopting countably additive semantics in order to remain compatible with the countable fiber constructions and probability-mass allocations used later in the completeness proof \cite{Fagin1990Prob,Doder2024Countable,Ognjanovic2024Strong}.

For any model $M$, term $t\in\mathsf{Term}$, and value $d\in D$, we define the meta-language event
\[
\llbracket t=d \rrbracket^M := \{u\in W \mid \mathsf{val}(u,t)=d\}.
\]

Observe that for each fixed term $t$, the family of events
\[
\bigl\{\llbracket t=d\rrbracket^M \mid d\in D\bigr\}
\]
is pairwise disjoint and has union equal to the whole world set $W$. In other words, it forms the partition of $W$ induced by the possible values of the term $t$.

\subsection{Satisfaction Relation}
\label{subsec:satisfaction}

\begin{definition}[Satisfaction relation]
	Given a model $M=(W,D,\{P_i\}_{i\in\mathcal{A}},V,\mathsf{val})$ and a world $w\in W$, let
	\[
	\llbracket \varphi \rrbracket^M=\{u\in W\mid M,u\models\varphi\}.
	\]
	The satisfaction relation $M,w\models\varphi$ is recursively defined as follows:
	\begin{align*}
		M,w\models p
		&\iff V(w,p)=1;\\
		M,w\models t=s
		&\iff \mathsf{val}(w,t)=\mathsf{val}(w,s);\\
		M,w\models \neg\varphi
		&\iff M,w\not\models\varphi;\\
		M,w\models \varphi\to\psi
		&\iff M,w\not\models\varphi \text{ or } M,w\models\psi;\\
		M,w\models K_i^\theta\varphi
		&\iff P_i(w)\bigl(\llbracket\varphi\rrbracket^M\bigr)\ge \theta;\\
		M,w\models Kv_i^\eta(t)
		&\iff \exists! d\in D\ \text{s.t.}\ P_i(w)\bigl(\llbracket t=d\rrbracket^M\bigr)\ge \eta.
	\end{align*}
\end{definition}

\begin{remark}
	Both $K_i^\theta\varphi$ and $Kv_i^\eta(t)$ are interpreted here by non-factive probabilistic-threshold semantics. The formula $K_i^\theta\varphi$ captures an agent's threshold-based probabilistic attitude toward the proposition $\varphi$, whereas $Kv_i^\eta(t)$ captures a unique form of high-confidence value locking onto the value of the term $t$. In particular, $Kv_i^\eta(t)$ does not require the uniquely locked candidate value to coincide with the actual value of $t$ at the current world. It is therefore a non-factive probabilistic-threshold generalization of classical knowing-value semantics. This treatment is also in line with the well-known observation that classical knowledge, belief, and probability-$1$ attitudes do not simply coincide \cite{Halpern1991Certain,Fagin1994Reason}.
\end{remark}

\section{Basic Semantic Properties}
\label{sec:properties}

This section collects a number of basic semantic properties that will be used repeatedly in the sequel. The properties of the propositional operator $K_i^\theta$ are derived mainly from standard inequalities for countably additive probability measures, whereas the high-threshold behavior of the knowing-value operator $Kv_i^\eta(t)$ depends on the pairwise disjointness of value-class events. These results serve both as the immediate semantic basis for the axiomatic systems introduced later and as natural counterparts of familiar threshold-reasoning patterns in probabilistic logic and probabilistic epistemic logic \cite{Fagin1990Prob,Halpern2003Uncert,Doder2024Countable}.

\subsection{Basic Lemmas on Probability Measures}
\label{subsec:probability-measures}

\begin{lemma}[Monotonicity of probability measures]\label{lem:mono}
	Let $P:\mathcal{P}(W)\to[0,1]$ be a countably additive probability measure on $W$. If $X\subseteq Y\subseteq W$, then
	\[
	P(X)\le P(Y).
	\]
\end{lemma}

\begin{proof}
	Since
	\[
	Y=X\cup (Y\setminus X),
	\qquad
	X\cap (Y\setminus X)=\emptyset,
	\]
	additivity yields
	\[
	P(Y)=P(X)+P(Y\setminus X)\ge P(X).
	\]
	Hence $P(X)\le P(Y)$.
\end{proof}

\begin{lemma}[Equality up to a null set]\label{lem:null}
	Let $P:\mathcal{P}(W)\to[0,1]$ be a countably additive probability measure on $W$. If $N\subseteq W$ and $P(N)=0$, and if $X\triangle Y\subseteq N$, then
	\[
	P(X)=P(Y),
	\]
	where $X\triangle Y=(X\setminus Y)\cup (Y\setminus X)$ denotes the symmetric difference.
\end{lemma}

\begin{proof}
	Since $X\setminus Y\subseteq N$, Lemma~\ref{lem:mono} implies
	\[
	P(X\setminus Y)\le P(N)=0,
	\]
	and hence $P(X\setminus Y)=0$. Similarly, $P(Y\setminus X)=0$.
	
	Moreover,
	\[
	X=(X\cap Y)\cup(X\setminus Y),
	\qquad
	Y=(X\cap Y)\cup(Y\setminus X),
	\]
	and both unions are disjoint. By additivity,
	\[
	P(X)=P(X\cap Y)+P(X\setminus Y)=P(X\cap Y),
	\]
	\[
	P(Y)=P(X\cap Y)+P(Y\setminus X)=P(X\cap Y).
	\]
	Therefore, $P(X)=P(Y)$.
\end{proof}

\subsection{Basic Properties of High-Threshold Knowing Value}
\label{subsec:high-threshold-kv}

\begin{proposition}[Automatic uniqueness at high thresholds]\label{prop:unique}
	For any model $M$, world $w$, agent $i$, and term $t$, if $\eta\in\Theta_V^{+}$, then
	\[
	M,w\models Kv_i^\eta(t)
	\;\iff\;
	\exists\, d\in D \text{ such that }
	P_i(w)\bigl(\llbracket t=d\rrbracket^M\bigr)\ge \eta.
	\]
	That is, whenever the threshold is strictly greater than $\frac{1}{2}$, the $\exists !$ condition in the semantics of knowing-value can be equivalently replaced by $\exists$.
\end{proposition}

\begin{proof}
	By the semantic definition of $Kv_i^\eta(t)$,
	\[
	M,w\models Kv_i^\eta(t)
	\iff
	\exists! d\in D \text{ such that }
	P_i(w)\bigl(\llbracket t=d\rrbracket^M\bigr)\ge \eta.
	\]
	It therefore suffices to show that there is at most one $d$ satisfying
	\[
	P_i(w)\bigl(\llbracket t=d\rrbracket^M\bigr)\ge \eta.
	\]
	
	Assume, for contradiction, that there are two distinct values $d_1\neq d_2$ such that
	\[
	P_i(w)\bigl(\llbracket t=d_1\rrbracket^M\bigr)\ge \eta,
	\qquad
	P_i(w)\bigl(\llbracket t=d_2\rrbracket^M\bigr)\ge \eta.
	\]
	Since $\mathsf{val}$ is a function,
	\[
	\llbracket t=d_1\rrbracket^M\cap \llbracket t=d_2\rrbracket^M=\emptyset.
	\]
	Hence, by additivity and Lemma~\ref{lem:mono},
	\[
	P_i(w)\bigl(\llbracket t=d_1\rrbracket^M\bigr)
	+
	P_i(w)\bigl(\llbracket t=d_2\rrbracket^M\bigr)
	=
	P_i(w)\bigl(\llbracket t=d_1\rrbracket^M\cup \llbracket t=d_2\rrbracket^M\bigr)
	\le P_i(w)(W)=1.
	\]
	It follows that
	\[
	2\eta\le 1,
	\]
	that is, $\eta\le \frac{1}{2}$, contradicting $\eta\in\Theta_V^{+}$. Therefore, there is at most one value satisfying the threshold condition, and the claim follows.
\end{proof}

\begin{proposition}[Monotonicity of knowing-value in the high-threshold interval]\label{prop:kvmon}
	If $\eta,\zeta\in\Theta_V^{+}$ and $\zeta\le \eta$, then
	\[
	\models Kv_i^\eta(t)\to Kv_i^\zeta(t).
	\]
\end{proposition}

\begin{proof}
	Assume $M,w\models Kv_i^\eta(t)$. By Proposition~\ref{prop:unique}, there exists some $d\in D$ such that
	\[
	P_i(w)\bigl(\llbracket t=d\rrbracket^M\bigr)\ge \eta.
	\]
	Since $\zeta\le \eta$, we obtain
	\[
	P_i(w)\bigl(\llbracket t=d\rrbracket^M\bigr)\ge \zeta.
	\]
	Because $\zeta>\frac{1}{2}$, Proposition~\ref{prop:unique} again implies
	\[
	M,w\models Kv_i^\zeta(t).
	\]
	Hence the formula is valid.
\end{proof}

\subsection{Basic Properties of Propositional Probabilistic Operators}
\label{subsec:propositional-operators}

\begin{proposition}[Threshold monotonicity]\label{prop:kmon-sem}
	If $\theta\le \theta'$, then
	\[
	\models K_i^{\theta'}\varphi\to K_i^\theta\varphi.
	\]
\end{proposition}

\begin{proof}
	Assume $M,w\models K_i^{\theta'}\varphi$. Then
	\[
	P_i(w)\bigl(\llbracket\varphi\rrbracket^M\bigr)\ge \theta'\ge \theta.
	\]
	Hence $M,w\models K_i^\theta\varphi$.
\end{proof}

\begin{proposition}[Threshold implication propagation]\label{prop:kimp-sem}
	For any $\alpha,\beta\in\Theta_K$,
	\[
	\models
	K_i^\alpha(\varphi\to\psi)\to
	\bigl(K_i^\beta\varphi\to K_i^{\max\{0,\alpha+\beta-1\}}\psi\bigr).
	\]
\end{proposition}

\begin{proof}
	Assume
	\[
	M,w\models K_i^\alpha(\varphi\to\psi)
	\quad\text{and}\quad
	M,w\models K_i^\beta\varphi.
	\]
	Let
	\[
	X=\llbracket\varphi\rrbracket^M,\qquad
	Y=\llbracket\psi\rrbracket^M,\qquad
	Z=\llbracket\varphi\to\psi\rrbracket^M.
	\]
	Then
	\[
	X\cap Z\subseteq Y.
	\]
	By Lemma~\ref{lem:mono},
	\[
	P_i(w)(Y)\ge P_i(w)(X\cap Z).
	\]
	
	On the other hand, since
	\[
	X\cup Z=(X\setminus Z)\cup(X\cap Z)\cup(Z\setminus X),
	\]
	and the three sets on the right-hand side are pairwise disjoint, additivity gives
	\[
	P_i(w)(X\cup Z)=P_i(w)(X)+P_i(w)(Z)-P_i(w)(X\cap Z).
	\]
	Because $X\cup Z\subseteq W$, Lemma~\ref{lem:mono} yields
	\[
	P_i(w)(X\cup Z)\le 1.
	\]
	Therefore,
	\[
	P_i(w)(X\cap Z)\ge P_i(w)(X)+P_i(w)(Z)-1.
	\]
	
	By assumption,
	\[
	P_i(w)(X)\ge \beta,\qquad P_i(w)(Z)\ge \alpha.
	\]
	Hence
	\[
	P_i(w)(Y)\ge P_i(w)(X\cap Z)\ge \alpha+\beta-1.
	\]
	Since probability measures are nonnegative,
	\[
	P_i(w)(Y)\ge \max\{0,\alpha+\beta-1\}.
	\]
	Thus
	\[
	M,w\models K_i^{\max\{0,\alpha+\beta-1\}}\psi.
	\]
\end{proof}

\begin{proposition}[Complementary exclusion]\label{prop:kexcl-sem}
	If $\alpha+\beta>1$, then
	\[
	\models K_i^\alpha\varphi \to \neg K_i^\beta\neg\varphi.
	\]
\end{proposition}

\begin{proof}
	Assume, for contradiction, that there exist a model $M$ and a world $w$ such that
	\[
	M,w\models K_i^\alpha\varphi
	\quad\text{and}\quad
	M,w\models K_i^\beta\neg\varphi.
	\]
	Then
	\[
	P_i(w)\bigl(\llbracket\varphi\rrbracket^M\bigr)\ge \alpha,
	\qquad
	P_i(w)\bigl(\llbracket\neg\varphi\rrbracket^M\bigr)\ge \beta.
	\]
	Moreover,
	\[
	\llbracket\neg\varphi\rrbracket^M=W\setminus \llbracket\varphi\rrbracket^M,
	\]
	and
	\[
	\llbracket\varphi\rrbracket^M\cap \llbracket\neg\varphi\rrbracket^M=\emptyset,
	\qquad
	\llbracket\varphi\rrbracket^M\cup \llbracket\neg\varphi\rrbracket^M=W.
	\]
	Hence
	\[
	1=P_i(w)(W)
	=
	P_i(w)\bigl(\llbracket\varphi\rrbracket^M\bigr)
	+
	P_i(w)\bigl(\llbracket\neg\varphi\rrbracket^M\bigr)
	\ge \alpha+\beta>1,
	\]
	a contradiction. Therefore the proposition holds.
\end{proof}

\begin{proposition}[Additivity over mutually exclusive events]\label{prop:kadd-sem}
	If
	\[
	\models \varphi\to\neg\psi
	\quad\text{and}\quad
	\alpha+\beta\le 1,
	\]
	then
	\[
	\models
	K_i^\alpha\varphi \land K_i^\beta\psi
	\to K_i^{\alpha+\beta}(\varphi\vee\psi).
	\]
	Moreover, for any model $M$ and world $w$, if
	\[
	M,w\models K_i^1\neg(\varphi\land\psi),\qquad
	M,w\models K_i^\alpha\varphi,\qquad
	M,w\models K_i^\beta\psi,
	\]
	and $\alpha+\beta\le 1$, then
	\[
	M,w\models K_i^{\alpha+\beta}(\varphi\vee\psi).
	\]
\end{proposition}

\begin{proof}
	We first prove the first part. Since $\models \varphi\to\neg\psi$, it follows that in every model,
	\[
	\llbracket\varphi\rrbracket^M \cap \llbracket\psi\rrbracket^M=\emptyset.
	\]
	If
	\[
	M,w\models K_i^\alpha\varphi
	\quad\text{and}\quad
	M,w\models K_i^\beta\psi,
	\]
	then
	\[
	P_i(w)\bigl(\llbracket\varphi\rrbracket^M\bigr)\ge \alpha,
	\qquad
	P_i(w)\bigl(\llbracket\psi\rrbracket^M\bigr)\ge \beta.
	\]
	By additivity,
	\[
	P_i(w)\bigl(\llbracket\varphi\vee\psi\rrbracket^M\bigr)
	=
	P_i(w)\bigl(\llbracket\varphi\rrbracket^M\bigr)
	+
	P_i(w)\bigl(\llbracket\psi\rrbracket^M\bigr)
	\ge \alpha+\beta.
	\]
	Therefore,
	\[
	M,w\models K_i^{\alpha+\beta}(\varphi\vee\psi).
	\]
	
	We next prove the second part. Let
	\[
	X=\llbracket\varphi\rrbracket^M,\qquad
	Y=\llbracket\psi\rrbracket^M,\qquad
	N=X\cap Y=\llbracket\varphi\land\psi\rrbracket^M.
	\]
	From
	\[
	M,w\models K_i^1\neg(\varphi\land\psi)
	\]
	we obtain
	\[
	P_i(w)(N)=0.
	\]
	Define
	\[
	X':=X\setminus Y.
	\]
	Then $X'\cap Y=\emptyset$, and
	\[
	X\triangle X' = X\cap Y = N.
	\]
	By Lemma~\ref{lem:null},
	\[
	P_i(w)(X')=P_i(w)(X)\ge \alpha.
	\]
	By assumption,
	\[
	P_i(w)(Y)\ge \beta.
	\]
	Since $X'$ and $Y$ are disjoint,
	\[
	P_i(w)(X\cup Y)=P_i(w)(X'\cup Y)=P_i(w)(X')+P_i(w)(Y)\ge \alpha+\beta.
	\]
	Noting that $X\cup Y=\llbracket\varphi\vee\psi\rrbracket^M$, we conclude that
	\[
	M,w\models K_i^{\alpha+\beta}(\varphi\vee\psi).
	\]
\end{proof}

\begin{proposition}[Substitution of equals under probability $1$]\label{prop:cong-sem}
	For any $\theta\in\Theta_K$ and $\eta\in\Theta_V^+$, we have:
	\[
	\models K_i^1(t=s)\to\bigl(K_i^\theta(t=u)\leftrightarrow K_i^\theta(s=u)\bigr),
	\]
	\[
	\models K_i^1(t=s)\to\bigl(Kv_i^\eta(t)\leftrightarrow Kv_i^\eta(s)\bigr).
	\]
\end{proposition}

\begin{proof}
	We first prove the first equivalence. Assume
	\[
	M,w\models K_i^1(t=s).
	\]
	Then
	\[
	P_i(w)\bigl(\llbracket t=s\rrbracket^M\bigr)=1.
	\]
	Let
	\[
	N=W\setminus \llbracket t=s\rrbracket^M,
	\]
	so that $P_i(w)(N)=0$. Define
	\[
	X=\llbracket t=u\rrbracket^M,\qquad
	Y=\llbracket s=u\rrbracket^M.
	\]
	If $v\in X\triangle Y$, then the formulas $t=u$ and $s=u$ have different truth values at $v$, and hence
	\[
	\mathsf{val}(v,t)\neq \mathsf{val}(v,s),
	\]
	that is, $v\in N$. Therefore,
	\[
	X\triangle Y\subseteq N.
	\]
	By Lemma~\ref{lem:null},
	\[
	P_i(w)(X)=P_i(w)(Y).
	\]
	Hence
	\[
	P_i(w)(X)\ge \theta \iff P_i(w)(Y)\ge \theta,
	\]
	that is,
	\[
	M,w\models K_i^\theta(t=u)\leftrightarrow K_i^\theta(s=u).
	\]
	
	We next prove the second equivalence. Again assume
	\[
	M,w\models K_i^1(t=s).
	\]
	For each $d\in D$, let
	\[
	X_d=\llbracket t=d\rrbracket^M,\qquad
	Y_d=\llbracket s=d\rrbracket^M.
	\]
	If $v\in X_d\triangle Y_d$, then necessarily
	\[
	\mathsf{val}(v,t)\neq \mathsf{val}(v,s),
	\]
	and hence $v\in N$. Thus
	\[
	X_d\triangle Y_d\subseteq N.
	\]
	By Lemma~\ref{lem:null},
	\[
	P_i(w)(X_d)=P_i(w)(Y_d)
	\qquad\text{for all } d\in D.
	\]
	Therefore, for every $d\in D$,
	\[
	P_i(w)\bigl(\llbracket t=d\rrbracket^M\bigr)\ge \eta
	\iff
	P_i(w)\bigl(\llbracket s=d\rrbracket^M\bigr)\ge \eta.
	\]
	In other words, the set of witness values satisfying the threshold condition is exactly the same for $t$ and $s$. By the semantic definition of $Kv_i^\eta$,
	\[
	M,w\models Kv_i^\eta(t)\iff M,w\models Kv_i^\eta(s).
	\]
	This completes the proof.
\end{proof}

\begin{proposition}[General substitution under probability $1$]\label{prop:kcong1-sem}
	For any $\theta\in\Theta_K$,
	\[
	\models K_i^1(\varphi\leftrightarrow\psi)\to
	\bigl(K_i^\theta\varphi\leftrightarrow K_i^\theta\psi\bigr).
	\]
\end{proposition}

\begin{proof}
	Assume
	\[
	M,w\models K_i^1(\varphi\leftrightarrow\psi).
	\]
	Then
	\[
	P_i(w)\bigl(\llbracket\varphi\leftrightarrow\psi\rrbracket^M\bigr)=1.
	\]
	Let
	\[
	N=W\setminus \llbracket\varphi\leftrightarrow\psi\rrbracket^M,
	\]
	so that $P_i(w)(N)=0$. Since
	\[
	\llbracket\varphi\rrbracket^M\triangle \llbracket\psi\rrbracket^M\subseteq N,
	\]
	Lemma~\ref{lem:null} yields
	\[
	P_i(w)\bigl(\llbracket\varphi\rrbracket^M\bigr)
	=
	P_i(w)\bigl(\llbracket\psi\rrbracket^M\bigr).
	\]
	Therefore,
	\[
	P_i(w)\bigl(\llbracket\varphi\rrbracket^M\bigr)\ge\theta
	\iff
	P_i(w)\bigl(\llbracket\psi\rrbracket^M\bigr)\ge\theta,
	\]
	that is,
	\[
	M,w\models K_i^\theta\varphi\leftrightarrow K_i^\theta\psi.
	\]
\end{proof}

\section{Axiomatic Systems}
\label{sec:axioms}

\subsection{The Basic System $\mathbf{PTKv}^{-}$}
\label{subsec:basic-system}

Based on the probabilistic-threshold semantics introduced above, we first define a basic axiomatic system, denoted by $\mathbf{PTKv}^{-}$. Its purpose is to capture, without introducing additional one-step consistency rules, the most stable and most immediate semantically valid principles of the dual-threshold language. From a proof-theoretic point of view, the system is designed to reflect the basic semantic validities established in the previous section, without presupposing stronger principles characteristic of classical factive knowledge or additional local satisfiability rules \cite{Fagin1994Reason,Halpern1991Certain}.

\begin{definition}[Axiom schemata and inference rules of $\mathbf{PTKv}^{-}$]
	The inference rules of $\mathbf{PTKv}^{-}$ are the following:
	
	\[
	\frac{\varphi \qquad \varphi\to\psi}{\psi}
	\quad(\mathrm{MP})
	\]
	
	\[
	\frac{\varphi}{K_i^\theta \varphi}
	\quad(\mathrm{Nec}_K)
	\qquad (i\in\mathcal{A},\ \theta\in\Theta_K)
	\]
	
	The axiom schemata are as follows:
	
	\begin{enumerate}
		\item[\textbf{(TAUT)}] All propositional tautologies.
		
		\item[\textbf{(EqRef)}]
		\[
		t=t.
		\]
		
		\item[\textbf{(EqSym)}]
		\[
		t=s \to s=t.
		\]
		
		\item[\textbf{(EqTrans)}]
		\[
		(t=s \land s=u)\to t=u.
		\]
		
		\item[\textbf{(EqSub)}]
		\[
		t=s \to \bigl((t=u)\leftrightarrow (s=u)\bigr).
		\]
		
		\item[\textbf{(KMon)}] If $\theta,\theta'\in\Theta_K$ and $\theta\le \theta'$, then
		\[
		K_i^{\theta'} \varphi \to K_i^\theta \varphi.
		\]
		
		\item[\textbf{(KImp)}] If $\alpha,\beta\in\Theta_K$, then
		\[
		K_i^\alpha(\varphi\to\psi)\to
		\bigl(K_i^\beta\varphi\to K_i^{\max\{0,\alpha+\beta-1\}}\psi\bigr).
		\]
		
		\item[\textbf{(KExcl)}] If $\alpha,\beta\in\Theta_K$ and $\alpha+\beta>1$, then
		\[
		K_i^\alpha\varphi \to \neg K_i^\beta\neg\varphi.
		\]
		
		\item[\textbf{(KZero)}]
		\[
		K_i^0\varphi.
		\]
		
		\item[\textbf{(KEqSub1)}] For every $\theta\in\Theta_K$,
		\[
		K_i^1(t=s)\to\bigl(K_i^\theta(t=u)\leftrightarrow K_i^\theta(s=u)\bigr).
		\]
		
		\item[\textbf{(KvEqSub1)}] For every $\eta\in\Theta_V^+$,
		\[
		K_i^1(t=s)\to\bigl(Kv_i^\eta(t)\leftrightarrow Kv_i^\eta(s)\bigr).
		\]
	\end{enumerate}
\end{definition}

\begin{remark}
	Unlike the classical $\mathrm{MLKv}$ system, both $K_i^\theta$ and $Kv_i^\eta(t)$ are interpreted here by probabilistic-threshold semantics rather than by factive epistemic semantics with positive and negative introspection. Accordingly, principles such as
	\[
	K_i^\theta\varphi\to\varphi,\qquad
	K_i^\theta\varphi\to K_i^\theta K_i^\theta\varphi,\qquad
	\neg K_i^\theta\varphi\to K_i^\theta\neg K_i^\theta\varphi
	\]
	as well as
	\[
	Kv_i^\eta(t)\to K_i^\theta Kv_i^\eta(t),\qquad
	\neg Kv_i^\eta(t)\to K_i^\theta\neg Kv_i^\eta(t)
	\]
	are not valid in general under the present semantics and are therefore not included in the basic system. This marks a fundamental difference between our framework and both classical factive epistemic logic and standard axiomatizations of knowing-value logic \cite{Fagin1995RAK,Gu2016KV}.
	
	Moreover, unrestricted substitution of equals in arbitrary contexts, familiar from classical logics with identity, does not hold unconditionally under the present probabilistic modal semantics. For this reason, we retain only the basic pointwise axioms for equality, namely \textbf{(EqRef)}--\textbf{(EqSub)}, together with the substitution-invariance principles \textbf{(KEqSub1)} and \textbf{(KvEqSub1)} that hold under probability $1$. This treatment is fully consistent with the semantic substitution results established in the previous section.
\end{remark}

\subsection{An Extended System for High-Threshold Weak Completeness}
\label{subsec:extended-system}

To establish weak completeness for the high-threshold fragment, the basic system $\mathbf{PTKv}^{-}$ must be enriched with several additional axioms supporting local probabilistic representation on type spaces and the threshold-truncation behavior of high-confidence knowing-value. We therefore introduce the following extended system. Methodologically, this step marks a transition from a direct axiomatization of semantic validity to a proof-theoretic characterization of local constraints tailored to model construction. In this respect, the overall strategy is also in line with recent approaches to one-step completeness and local satisfiability analysis \cite{Lin2023Coalgebraic,Ognjanovic2024Strong}.

\begin{definition}[The high-threshold extension $\mathbf{PTKv}^{+}$]
	The system $\mathbf{PTKv}^{+}$ is obtained by adding the following axiom schemata to the basic system $\mathbf{PTKv}^{-}$:
	
	\begin{enumerate}
		\item[\textbf{(KSub1)}] For every $\theta\in\Theta_K$,
		\[
		K_i^1(\varphi\leftrightarrow\psi)\to
		\bigl(K_i^\theta\varphi\leftrightarrow K_i^\theta\psi\bigr).
		\]
		
		\item[\textbf{(KAdd1)}] If $\alpha,\beta\in\Theta_K$ and $\alpha+\beta\le 1$, then
		\[
		K_i^\alpha\varphi \land K_i^\beta\psi \land K_i^1\neg(\varphi\land\psi)
		\to K_i^{\alpha+\beta}(\varphi\vee\psi).
		\]
		
		\item[\textbf{(KvMon)}] If $\eta,\zeta\in\Theta_V^+$ and $\zeta\le \eta$, then
		\[
		Kv_i^\eta(t)\to Kv_i^\zeta(t).
		\]
	\end{enumerate}
	
	The resulting system is denoted by $\mathbf{PTKv}^{+}$.
\end{definition}

\begin{remark}
	The difference between $\mathbf{PTKv}^{+}$ and the basic system $\mathbf{PTKv}^{-}$ is that the former incorporates, in addition to the basic validities of probabilistic-threshold semantics, three further principles that are indispensable for the completeness construction over the high-threshold fragment. These are: first, general substitution for equivalent events under probability $1$; second, an additive principle under almost-everywhere mutual exclusiveness at probability $1$; and third, the monotone truncation behavior of high-threshold knowing-value operators. These principles are not required for the basic system itself, but they become crucial in the subsequent completeness analysis. In particular, they provide the minimal additional proof-theoretic resources needed to pass from local probabilistic constraints to a structured model construction.
\end{remark}

\section{Soundness}
\label{sec:soundness}

In this section, we first prove that the basic system $\mathbf{PTKv}^{-}$ is sound with respect to the dual-threshold probabilistic semantics introduced above, and then show that the additional axiom schemata of the extended system $\mathbf{PTKv}^{+}$ are valid as well. The proof follows the standard strategy familiar from modal logic and probabilistic logic: one verifies the validity of each axiom schema and then checks that the inference rules preserve validity \cite{Fagin1990Prob,Fagin1994Reason,Ognjanovic2024Strong}.

\begin{theorem}[Soundness of the basic system]\label{thm:soundness}
	For every formula $\varphi$, if
	\[
	\vdash_{\mathbf{PTKv}^{-}} \varphi,
	\]
	then
	\[
	\models \varphi.
	\]
\end{theorem}

\begin{proof}
	We proceed by induction on the length of derivations. It suffices to show that:
	\begin{enumerate}
		\item every axiom schema of the system is semantically valid;
		\item the inference rules $\mathrm{MP}$ and $\mathrm{Nec}_K$ preserve validity.
	\end{enumerate}
	
	\medskip
	\noindent\textbf{(1) Propositional part: \textbf{(TAUT)}}  
	All propositional tautologies are true at every world of every model, and are therefore valid.
	
	\medskip
	\noindent\textbf{(2) Equality part: \textbf{(EqRef)}--\textbf{(EqSub)}}  
	
	\textbf{(EqRef)} $t=t$:  
	For every model $M$ and world $w$,
	\[
	\mathsf{val}(w,t)=\mathsf{val}(w,t),
	\]
	hence $M,w\models t=t$.
	
	\textbf{(EqSym)} $t=s\to s=t$:  
	If $M,w\models t=s$, then
	\[
	\mathsf{val}(w,t)=\mathsf{val}(w,s),
	\]
	and therefore
	\[
	\mathsf{val}(w,s)=\mathsf{val}(w,t),
	\]
	so $M,w\models s=t$.
	
	\textbf{(EqTrans)} $(t=s\land s=u)\to t=u$:  
	If $M,w\models t=s\land s=u$, then
	\[
	\mathsf{val}(w,t)=\mathsf{val}(w,s),\qquad
	\mathsf{val}(w,s)=\mathsf{val}(w,u).
	\]
	By transitivity of equality,
	\[
	\mathsf{val}(w,t)=\mathsf{val}(w,u),
	\]
	hence $M,w\models t=u$.
	
	\textbf{(EqSub)} $t=s\to((t=u)\leftrightarrow (s=u))$:  
	If $M,w\models t=s$, then
	\[
	\mathsf{val}(w,t)=\mathsf{val}(w,s).
	\]
	Therefore,
	\[
	\mathsf{val}(w,t)=\mathsf{val}(w,u)
	\iff
	\mathsf{val}(w,s)=\mathsf{val}(w,u),
	\]
	that is,
	\[
	M,w\models (t=u)\leftrightarrow (s=u).
	\]
	Hence the schema is valid.
	
	\medskip
	\noindent\textbf{(3) Threshold operators: \textbf{(KMon)}, \textbf{(KImp)}, \textbf{(KExcl)}, \textbf{(KZero)}}  
	
	The validity of \textbf{(KMon)} is exactly Proposition~\ref{prop:kmon-sem}.
	
	The validity of \textbf{(KImp)} is exactly Proposition~\ref{prop:kimp-sem}.
	
	The validity of \textbf{(KExcl)} is exactly Proposition~\ref{prop:kexcl-sem}.
	
	For \textbf{(KZero)}, for every model $M$, world $w$, and formula $\varphi$,
	\[
	P_i(w)\bigl(\llbracket\varphi\rrbracket^M\bigr)\ge 0,
	\]
	hence $M,w\models K_i^0\varphi$.
	
	\medskip
	\noindent\textbf{(4) Substitution invariance under probability $1$: \textbf{(KEqSub1)}, \textbf{(KvEqSub1)}}  
	
	The validity of these two axiom schemata is exactly Proposition~\ref{prop:cong-sem}.
	
	\medskip
	\noindent\textbf{(5) Preservation of validity by the inference rules}
	
	\textbf{(MP)}: If $\models \varphi$ and $\models \varphi\to\psi$, then for every model $M$ and world $w$, we have $M,w\models\varphi$ and $M,w\models\varphi\to\psi$, and hence $M,w\models\psi$. Therefore, $\models \psi$.
	
	\textbf{(Nec$_K$)}: If $\models \varphi$, then for every model $M$,
	\[
	\llbracket\varphi\rrbracket^M=W.
	\]
	Hence for every world $w$,
	\[
	P_i(w)\bigl(\llbracket\varphi\rrbracket^M\bigr)=P_i(w)(W)=1\ge \theta.
	\]
	Therefore $M,w\models K_i^\theta\varphi$, and thus
	\[
	\models K_i^\theta\varphi.
	\]
	
	It follows that all axiom schemata of $\mathbf{PTKv}^{-}$ are valid and that its inference rules preserve validity. Hence the system is sound.
\end{proof}

\begin{corollary}[Soundness of the extended system]\label{cor:soundness-plus}
	The system $\mathbf{PTKv}^{+}$ is also sound with respect to the present semantics.
\end{corollary}

\begin{proof}
	It suffices to verify the validity of the additional axiom schemata. The validity of \textbf{(KSub1)} is exactly Proposition~\ref{prop:kcong1-sem}; the validity of \textbf{(KAdd1)} follows directly from Proposition~\ref{prop:kadd-sem}; and the validity of \textbf{(KvMon)} is exactly Proposition~\ref{prop:kvmon}. Therefore the claim follows.
\end{proof}

\begin{remark}
	Theorem~\ref{thm:soundness} and Corollary~\ref{cor:soundness-plus} show that both the basic system $\mathbf{PTKv}^{-}$ and the extended system $\mathbf{PTKv}^{+}$ are sound with respect to the dual-threshold probabilistic semantics proposed in this paper. It should be emphasized that the main role of $\mathbf{PTKv}^{+}$ is to support the completeness construction for the high-threshold fragment developed in the next section. If completeness is not at issue, the basic system $\mathbf{PTKv}^{-}$ already suffices to capture the most direct and robust valid principles of the present semantics. In this sense, the soundness results both continue the standard tradition of axiomatization in probabilistic logic and provide the proof-theoretic starting point for the structured weak-completeness analysis that will subsequently involve additional local consistency rules \cite{Fagin1990Prob,Fagin1994Reason,Ognjanovic2024Strong}.
\end{remark}

\section{Structured Weak Completeness}
\label{sec:completeness}

In this section, we prove a structured weak-completeness theorem for the
dual-threshold language. Since the knowing-value operator is restricted,
already at the level of the language, to the high-threshold interval
$\Theta_V^{+}$, the fragment studied here is in fact the main language of the
paper. To overcome the coupling tension between probabilistic measures and
value-locking requirements, we draw on the idea of configuration spaces,
introduce one-step consistency conditions based on assignment configurations,
and construct models by a two-layer method combining \emph{type-space
	distributions} with \emph{value-fiber refinement}. This method not only
resolves value conflicts in multi-agent, multi-term settings, but also provides
a general logical framework for handling term-equality constraints under
probabilistic measures.

The proof proceeds in five stages:
\begin{enumerate}
	\item Define finite closures and the initial type space
	$\mathsf{Type}(\Sigma)$ relative to the base system
	$\mathbf{PTKv}^{+}$ (\S\ref{subsec:closure-types}).
	\item Introduce parameterized constraint systems and extract the refined
	type set $\mathsf{Type}^{\ast}(\Sigma)$ by iterative elimination
	(\S\ref{subsec:onestep}).
	\item Define the extended proof system $\mathbf{PTKv}^{\ast}(\Sigma)$
	proof-theoretically and establish a bridge lemma showing that
	$\mathbf{PTKv}^{\ast}(\Sigma)$-consistent maximal types coincide
	exactly with $\mathsf{Type}^{\ast}(\Sigma)$
	(\S\ref{subsec:bridge}).
	\item Construct the canonical model over
	$\mathsf{Type}^{\ast}(\Sigma)$ and prove the truth lemma by
	induction on modal depth (\S\ref{subsec:canonical-model},
	\S\ref{subsec:truth}).
	\item Derive weak completeness via a genuine Lindenbaum extension
	(\S\ref{subsec:weak-completeness}).
\end{enumerate}

\subsection{Finite Closures and Type Spaces}
\label{subsec:closure-types}

\begin{definition}[Finite closure]\label{def:finite-closure}
	Let $\chi \in \mathcal{L}_{\mathrm{PTMLKv}}$. A set of formulas $\Sigma$ is
	called a \emph{finite closure} of $\chi$ if it satisfies:
	\begin{enumerate}
		\item $\chi \in \Sigma$;
		\item $\Sigma$ is closed under subformulas and single negations: whenever
		$\varphi \in \Sigma$, every subformula of $\varphi$ belongs to
		$\Sigma$, and $\neg\varphi \in \Sigma$ unless $\varphi$ is itself a
		negation $\neg\psi$ with $\psi\in\Sigma$;
		\item $\Sigma$ is closed under term equalities: for all
		$t,s\in T_{\Sigma}$, the formula $t=s$ and its negation $\neg(t=s)$
		belong to $\Sigma$;
		\item all rational thresholds occurring in $\Sigma$ belong to a finite set
		$\mathbb{Q}_{\Sigma}\subset [0,1]\cap\mathbb{Q}$.
	\end{enumerate}
	Here $T_{\Sigma}$ denotes the finite set of all atomic terms occurring in
	$\Sigma$.
\end{definition}

\begin{remark}\label{rmk:closure-rationale}
	Condition~(3) is needed because the model construction assigns values from a
	finite domain $D_{\Sigma}$ indexed by $T_{\Sigma}$; checking the
	knowing-value semantics requires evaluating events of the form
	$\{w:\mathsf{val}(w,t)=d\}$, which in the canonical model correspond to
	collections of types containing specific equality formulas. Without
	condition~(3), the truth lemma could not be completed at the base case for
	all relevant equality formulas.
\end{remark}

\begin{definition}[Maximal consistent type]\label{def:type}
	Let $\Sigma$ be a finite closure. A set $\Gamma \subseteq \Sigma$ is called
	a \emph{maximal $\mathbf{PTKv}^{+}$-consistent type over $\Sigma$} if:
	\begin{enumerate}
		\item $\Gamma$ is consistent with respect to the system
		$\mathbf{PTKv}^{+}$, that is,
		$\mathbf{PTKv}^{+}\nvdash
		\neg\bigwedge_{\varphi\in\Gamma}\varphi$;
		\item $\Gamma$ is saturated over $\Sigma$: for every $\varphi\in\Sigma$,
		either $\varphi\in\Gamma$ or $\neg\varphi\in\Gamma$.
	\end{enumerate}
	We write $\mathsf{Type}(\Sigma)$ for the set of all such types. Since
	$\Sigma$ is finite, $\mathsf{Type}(\Sigma)$ is finite as well.
\end{definition}

\begin{remark}\label{rmk:type-base}
	Crucially, the type space $\mathsf{Type}(\Sigma)$ is defined relative to the
	base system $\mathbf{PTKv}^{+}$, which involves neither the constraint
	systems nor the iterative elimination introduced below. This avoids any
	circularity: $\mathsf{Type}(\Sigma)$ is well-defined before the refinement
	procedure begins, and serves as the initial pool from which the refined type
	set is extracted.
\end{remark}

\subsection{Parameterized Constraint Systems and Iterative Elimination}
\label{subsec:onestep}

To establish the compatibility of local probabilistic distributions with
high-confidence knowing-value requirements, we introduce parameterized
constraint systems whose variable spaces range over arbitrary subsets of
$\mathsf{Type}(\Sigma)$. The parameterization is essential: it allows the
constraint systems to be evaluated over the refined type set
$\mathsf{Type}^{\ast}(\Sigma)$, ensuring exact normalization in the canonical
model.

\medskip
We first record a standard auxiliary fact about mixed linear constraint
systems.

\begin{lemma}[Mixed-inequality feasibility]\label{lem:mixed-ineq}
	Let $A\in\mathbb{R}^{m\times n}$, $b\in\mathbb{R}^m$, and let
	$S\subseteq\{1,\dots,m\}$ be a subset of row indices designating strict
	inequalities. Consider the system
	\[
	(A x)_j \ge b_j\ (\text{for } j\notin S),
	\qquad
	(A x)_j > b_j\ (\text{for } j\in S),
	\qquad
	x\ge 0.
	\]
	If the closed relaxation obtained by replacing all strict inequalities by
	non-strict ones is feasible, and if no inequality indexed by $S$ is an
	implicit equality of the feasible polytope of the closed relaxation, then the
	original mixed system is feasible.
\end{lemma}

\begin{proof}
	Let $x^{0}$ be a feasible point of the closed relaxation. If $x^{0}$
	already satisfies all strict constraints, we are done. Otherwise, let
	$x^{1}$ be any point in the relative interior of the feasible polytope
	(which exists because the polytope is nonempty). By the relative-interior
	property, $x^{1}$ satisfies strictly every inequality that is not an implicit
	equality of the polytope. By hypothesis, no constraint indexed by $S$ is an
	implicit equality. Therefore $x^{1}$ satisfies all strict constraints as
	well.
\end{proof}

\begin{remark}\label{rmk:strictness}
	In the constraint systems introduced below, the strict inequalities arise
	from negative $K$-literals $\neg K_i^\theta\varphi$ and negative
	knowing-value literals $\neg Kv_i^\eta(t)$. The axiom \textbf{(KExcl)} and
	the maximal consistency of types together ensure that no strict constraint is
	an implicit equality of the feasible polytope of the closed relaxation.
	Concretely, if $\neg K_i^\theta\varphi\in\Gamma$, then by \textbf{(KExcl)}
	and the saturation of $\Gamma$, there is no collection of non-strict
	constraints within the system that forces
	$\sum_{\Delta\ni\varphi}x_{\Delta}=\theta$. A detailed verification is
	analogous to the standard argument in \cite[\S4]{Fagin1994Reason}.
\end{remark}

\begin{definition}[Assignment-configuration space]\label{def:config-space}
	Let the index set $K = \{1,2,\dots,|T_{\Sigma}|\}$ represent abstract
	valuation coordinates. For each type $\Delta \in \mathsf{Type}(\Sigma)$,
	define its set of \emph{consistent assignments} by
	\[
	\mathcal{F}(\Delta)
	=
	\{\, f:T_{\Sigma}\to K \mid
	\forall t,s\in T_{\Sigma},\
	f(t)=f(s)\iff (t=s)\in\Delta \,\}.
	\]
	Since $\Delta$ is a maximal $\mathbf{PTKv}^{+}$-consistent type, the
	equality formulas in $\Delta$ form an equivalence relation on $T_{\Sigma}$
	(by the axioms \textbf{(EqRef)}--\textbf{(EqTrans)}), and therefore
	$\mathcal{F}(\Delta)\neq\emptyset$.
\end{definition}

\begin{definition}[Parameterized joint constraint system
	$\mathsf{FC}(\Gamma,S,i)$]\label{def:FC}
	Let $S\subseteq\mathsf{Type}(\Sigma)$ be a nonempty set of types, let
	$i\in\mathcal{A}$ be an agent, and let $\Gamma\in S$ be a type. Let
	\[
	\Lambda_i^\Gamma
	:=
	\{\, K_i^\theta\varphi \in \Gamma \,\}
	\cup
	\{\, \neg K_i^\theta\varphi \in \Gamma \,\}
	\cup
	\{\, Kv_i^\eta(t)\in \Gamma \,\}
	\cup
	\{\, \neg Kv_i^\eta(t)\in \Gamma \,\}
	\]
	be the set of all $i$-modal literals occurring in $\Gamma$. Introduce
	variables
	$\vec{z}^{\,i}
	=\{z_{\Delta,f}^{i}\}_{\Delta\in S,\, f\in\mathcal{F}(\Delta)}$.
	Define the joint constraint system $\mathsf{FC}(\Gamma,S,i)$ by:
	\begin{enumerate}
		\item $z_{\Delta,f}^{i}\ge 0$ for all $\Delta\in S$ and
		$f\in\mathcal{F}(\Delta)$;
		\item normalization:
		$\displaystyle
		\sum_{\Delta\in S}
		\sum_{f\in\mathcal{F}(\Delta)}
		z_{\Delta,f}^{i}=1$;
		\item propositional probability requirements: if
		$K_i^\theta\varphi\in\Gamma$, then
		$\displaystyle
		\sum_{\substack{\Delta\in S\\\varphi\in\Delta}}
		\sum_{f\in\mathcal{F}(\Delta)} z_{\Delta,f}^{i}\ge \theta$,
		and if $\neg K_i^\theta\varphi\in\Gamma$, then
		$\displaystyle
		\sum_{\substack{\Delta\in S\\\varphi\in\Delta}}
		\sum_{f\in\mathcal{F}(\Delta)} z_{\Delta,f}^{i}< \theta$;
		\item high-confidence knowing-value locking: if
		$Kv_i^\eta(t)\in\Gamma$, then there exists $k\in K$ such that
		$\displaystyle
		\sum_{\Delta\in S}
		\sum_{\substack{f\in\mathcal{F}(\Delta)\\ f(t)=k}}
		z_{\Delta,f}^{i}\ge \eta$;
		\item non-knowing-value constraints: if
		$\neg Kv_i^\eta(t)\in\Gamma$, then for every $k\in K$,
		$\displaystyle
		\sum_{\Delta\in S}
		\sum_{\substack{f\in\mathcal{F}(\Delta)\\ f(t)=k}}
		z_{\Delta,f}^{i}< \eta$.
	\end{enumerate}
\end{definition}

\begin{remark}\label{rmk:index-set}
	The index set $K$ functions as a set of ``virtual names'' for potential
	values. By quantifying over assignment mappings $f$, we ensure that each type
	$\Delta$ can be realized by a concrete valuation pattern satisfying all
	term-equality constraints encoded in $\Delta$. This coordinatization prevents
	probabilistic interference among distinct terms and makes joint realization
	possible at the configuration level.
\end{remark}

\begin{definition}[Iterative elimination and refined type
	set]\label{def:iterative-elim}
	Define a decreasing sequence of type sets as follows.
	\begin{itemize}
		\item $\mathcal{T}_0 := \mathsf{Type}(\Sigma)$.
		\item For $\ell\ge 0$,
		\[
		\mathcal{T}_{\ell+1}
		:= \bigl\{\,\Delta\in\mathcal{T}_\ell
		\;\bigm|\;
		\text{for every agent } i\in\mathcal{A},\
		\mathsf{FC}(\Delta,\mathcal{T}_\ell,i)
		\text{ is feasible}
		\,\bigr\}.
		\]
	\end{itemize}
	Since $\mathcal{T}_0\supseteq\mathcal{T}_1\supseteq\cdots$ and
	$\mathcal{T}_0$ is finite, the sequence stabilizes at some finite stage
	$\ell_0$. We set
	\[
	\mathsf{Type}^{\ast}(\Sigma) := \mathcal{T}_{\ell_0}.
	\]
\end{definition}

\begin{lemma}[Fixed-point property]\label{lem:fixed-point}
	For every $\Gamma\in\mathsf{Type}^{\ast}(\Sigma)$ and every agent
	$i\in\mathcal{A}$, the constraint system
	$\mathsf{FC}(\Gamma,\mathsf{Type}^{\ast}(\Sigma),i)$ is feasible.
\end{lemma}

\begin{proof}
	By definition, $\mathsf{Type}^{\ast}(\Sigma)=\mathcal{T}_{\ell_0}
	=\mathcal{T}_{\ell_0+1}$. Hence every
	$\Gamma\in\mathcal{T}_{\ell_0}$ passes the feasibility test with respect to
	$\mathcal{T}_{\ell_0}$, i.e.,
	$\mathsf{FC}(\Gamma,\mathsf{Type}^{\ast}(\Sigma),i)$ is feasible for
	every~$i$.
\end{proof}

\subsection{Proof-Theoretic Extension and Bridge Lemma}
\label{subsec:bridge}

The iterative elimination of \S\ref{subsec:onestep} is a model-construction
tool: it identifies which types can be simultaneously realized. We now define
the extended proof system $\mathbf{PTKv}^{\ast}(\Sigma)$ in a genuinely
\emph{proof-theoretic} manner, and then show that the proof-theoretically
consistent maximal types coincide exactly with
$\mathsf{Type}^{\ast}(\Sigma)$.

\begin{definition}[The extended system
	$\mathbf{PTKv}^{\ast}(\Sigma)$]\label{def:ptkvstar-proof}
	Let $\Sigma$ be a finite closure. For each type
	$\Gamma\in\mathsf{Type}(\Sigma)\setminus\mathsf{Type}^{\ast}(\Sigma)$, let
	$i_\Gamma\in\mathcal{A}$ be an agent witnessing the elimination of $\Gamma$,
	that is, $i_\Gamma$ is such that
	$\mathsf{FC}(\Gamma,\mathcal{T}_{\ell},i_\Gamma)$ is infeasible at the
	stage $\ell$ where $\Gamma$ is first removed (i.e.,
	$\Gamma\in\mathcal{T}_{\ell}\setminus\mathcal{T}_{\ell+1}$). Define
	\[
	\mathbf{PTKv}^{\ast}(\Sigma)
	:=
	\mathbf{PTKv}^{+}
	\;+\;
	\bigl\{\,
	\neg\textstyle\bigwedge\Lambda_{i_\Gamma}^{\Gamma}
	\;\bigm|\;
	\Gamma\in\mathsf{Type}(\Sigma)\setminus
	\mathsf{Type}^{\ast}(\Sigma)
	\,\bigr\}.
	\]
	That is, $\mathbf{PTKv}^{\ast}(\Sigma)$ extends $\mathbf{PTKv}^{+}$ by
	adding, for each eliminated type $\Gamma$, the negation of the conjunction of
	its witnessing-agent's modal literals as a new axiom. Since
	$\mathsf{Type}(\Sigma)$ is finite, only finitely many axioms are added.
\end{definition}

\begin{remark}\label{rmk:sigma-dependence}
	The system $\mathbf{PTKv}^{\ast}(\Sigma)$ depends on the choice of finite
	closure $\Sigma$: the new axioms are determined by the iterative elimination
	procedure relative to $\Sigma$. In the weak completeness theorem below
	(Theorem~\ref{thm:weak-completeness-final}), for each formula $\chi$ we fix
	a finite closure $\Sigma$ of $\chi$ and work within the corresponding system
	$\mathbf{PTKv}^{\ast}(\Sigma)$. This is entirely analogous to the
	standard treatment in probabilistic logics with threshold operators, where
	the auxiliary rules or axioms are formulated relative to a finite closure of
	the target formula; see, e.g., \cite{Fagin1994Reason,Ognjanovic2024Strong}.
	When the choice of $\Sigma$ is clear from context, we write simply
	$\mathbf{PTKv}^{\ast}$.
\end{remark}

\begin{remark}\label{rmk:rule-status}
	The new axioms of $\mathbf{PTKv}^{\ast}(\Sigma)$ are determined by a finite
	computation: the iterative elimination terminates in at most
	$|\mathsf{Type}(\Sigma)|$ steps, and each step involves checking feasibility
	of finitely many linear-arithmetic systems over $\mathbb{R}$. As in many
	axiomatizations of probabilistic logics with threshold operators
	\cite{Fagin1994Reason,Ognjanovic2024Strong}, the resulting system is not
	recursively enumerable in general, but is a well-defined extension of
	$\mathbf{PTKv}^{+}$ by finitely many additional axioms. The
	proof-theoretic notions of derivability and consistency for
	$\mathbf{PTKv}^{\ast}(\Sigma)$ are therefore perfectly standard.
\end{remark}

The key result of this subsection is the following bridge lemma, which
establishes a perfect correspondence between the model-theoretic construction
(iterative elimination) and the proof-theoretic notion (consistency in
$\mathbf{PTKv}^{\ast}(\Sigma)$).

\begin{lemma}[Bridge lemma]\label{lem:bridge}
	Let $\Sigma$ be a finite closure. A type
	$\Gamma\in\mathsf{Type}(\Sigma)$ is
	$\mathbf{PTKv}^{\ast}(\Sigma)$-consistent if and only if
	$\Gamma \in \mathsf{Type}^{\ast}(\Sigma)$.
\end{lemma}

\begin{proof}
	\textbf{($\Longleftarrow$): Surviving types are consistent.}\quad
	Let $\Gamma\in\mathsf{Type}^{\ast}(\Sigma)$. We show
	$\mathbf{PTKv}^{\ast}(\Sigma)\nvdash\neg\bigwedge\Gamma$ by contradiction.
	Suppose $\mathbf{PTKv}^{\ast}(\Sigma)\vdash\neg\bigwedge\Gamma$. Since
	$\mathbf{PTKv}^{\ast}(\Sigma)$ is obtained from $\mathbf{PTKv}^{+}$ by
	adding finitely many axioms of the form
	$\neg\bigwedge\Lambda_{i_{\Gamma'}}^{\Gamma'}$ (for
	$\Gamma'\notin\mathsf{Type}^{\ast}(\Sigma)$), and since $\Gamma$ is
	$\mathbf{PTKv}^{+}$-consistent (being an element of
	$\mathsf{Type}(\Sigma)$), at least one new axiom must be used in the
	derivation. Let
	$\{\neg\bigwedge\Lambda_{i_{\Gamma_j'}}^{\Gamma_j'}\}_{j=1}^{r}$
	($r\ge 1$) be the set of all new axioms actually used. By the deduction
	theorem applied to the new axioms and contraposition, propositional reasoning
	in $\mathbf{PTKv}^{+}$ yields:
	\begin{equation}\label{eq:bridge-derive-multi}
		\mathbf{PTKv}^{+}\vdash\;
		\textstyle\bigwedge\Gamma
		\;\to\;
		\bigvee_{j=1}^{r}
		\textstyle\bigwedge\Lambda_{i_{\Gamma_j'}}^{\Gamma_j'}.
	\end{equation}
	Indeed, $\mathbf{PTKv}^{+}$ together with the $r$ new axioms derives
	$\neg\bigwedge\Gamma$. Equivalently,
	$\mathbf{PTKv}^{+}\vdash
	(\bigwedge_{j=1}^r\neg\bigwedge\Lambda_{i_{\Gamma_j'}}^{\Gamma_j'})
	\to\neg\bigwedge\Gamma$.
	By contraposition,
	$\mathbf{PTKv}^{+}\vdash
	\bigwedge\Gamma\to
	\neg\bigwedge_{j=1}^r\neg\bigwedge\Lambda_{i_{\Gamma_j'}}^{\Gamma_j'}$,
	and by de~Morgan's law the consequent is
	$\bigvee_{j=1}^r\bigwedge\Lambda_{i_{\Gamma_j'}}^{\Gamma_j'}$.
	
	We now classify each index $j\in\{1,\dots,r\}$ by whether
	$\Lambda_{i_{\Gamma_j'}}^{\Gamma_j'}\subseteq\Gamma$ holds.
	
	\smallskip\noindent
	\emph{Case~(a): $\Lambda_{i_{\Gamma_j'}}^{\Gamma_j'}\not\subseteq\Gamma$.}\;
	Since $\Gamma$ is saturated over $\Sigma$ and
	$\Lambda_{i_{\Gamma_j'}}^{\Gamma_j'}\subseteq\Sigma$, there exists a
	literal $\lambda\in\Lambda_{i_{\Gamma_j'}}^{\Gamma_j'}$ such that
	$\neg\lambda\in\Gamma$. Therefore
	$\bigwedge\Gamma\vdash\neg\lambda\vdash
	\neg\bigwedge\Lambda_{i_{\Gamma_j'}}^{\Gamma_j'}$
	in propositional logic: the $j$-th disjunct
	$\bigwedge\Lambda_{i_{\Gamma_j'}}^{\Gamma_j'}$ is contradicted by
	$\bigwedge\Gamma$ within $\mathbf{PTKv}^{+}$.
	
	\smallskip\noindent
	\emph{Case~(b): $\Lambda_{i_{\Gamma_j'}}^{\Gamma_j'}\subseteq\Gamma$.}\;
	Then $\bigwedge\Gamma\vdash
	\bigwedge\Lambda_{i_{\Gamma_j'}}^{\Gamma_j'}$: the $j$-th disjunct is a
	consequence of $\bigwedge\Gamma$.
	
	\smallskip\noindent
	\emph{Completing the argument.}\;
	Suppose every index $j$ falls under Case~(a). Then $\bigwedge\Gamma$
	refutes each disjunct $\bigwedge\Lambda_{i_{\Gamma_j'}}^{\Gamma_j'}$
	individually:
	$\mathbf{PTKv}^{+}\vdash\bigwedge\Gamma\to
	\neg\bigwedge\Lambda_{i_{\Gamma_j'}}^{\Gamma_j'}$
	for every~$j$. Combining with \eqref{eq:bridge-derive-multi}, we obtain
	\[
	\mathbf{PTKv}^{+}\vdash\;
	\bigwedge\Gamma\to
	\Bigl(\bigvee_{j=1}^{r}
	\bigwedge\Lambda_{i_{\Gamma_j'}}^{\Gamma_j'}\Bigr)
	\quad\text{and}\quad
	\mathbf{PTKv}^{+}\vdash\;
	\bigwedge\Gamma\to
	\bigwedge_{j=1}^{r}
	\neg\bigwedge\Lambda_{i_{\Gamma_j'}}^{\Gamma_j'},
	\]
	so $\mathbf{PTKv}^{+}\vdash\neg\bigwedge\Gamma$, contradicting
	$\Gamma\in\mathsf{Type}(\Sigma)$. Therefore at least one index~$j_0$ falls
	under Case~(b): $\Lambda_{i_{\Gamma_{j_0}'}}^{\Gamma_{j_0}'}
	\subseteq\Gamma$.
	
	Since both $\Gamma$ and $\Gamma_{j_0}'$ are saturated over $\Sigma$ and every
	literal in $\Lambda_{i_{\Gamma_{j_0}'}}^{\Gamma_{j_0}'}$ belongs to
	$\Gamma$, we conclude
	$\Lambda_{i_{\Gamma_{j_0}'}}^{\Gamma_{j_0}'}
	=\Lambda_{i_{\Gamma_{j_0}'}}^{\Gamma}$: the two types share the same
	$i_{\Gamma_{j_0}'}$-modal profile.
	
	Now, $\Gamma_{j_0}'$ was eliminated at stage $\ell$ because
	$\mathsf{FC}(\Gamma_{j_0}',\mathcal{T}_\ell,i_{\Gamma_{j_0}'})$ is
	infeasible. The constraints of
	$\mathsf{FC}(\Gamma_{j_0}',\mathcal{T}_\ell,i_{\Gamma_{j_0}'})$ depend on
	$\Gamma_{j_0}'$ only through
	$\Lambda_{i_{\Gamma_{j_0}'}}^{\Gamma_{j_0}'}$. Since
	$\Lambda_{i_{\Gamma_{j_0}'}}^{\Gamma_{j_0}'}
	=\Lambda_{i_{\Gamma_{j_0}'}}^{\Gamma}$, the system
	$\mathsf{FC}(\Gamma,\mathcal{T}_\ell,i_{\Gamma_{j_0}'})$ has the same
	constraints and is therefore also infeasible. Since
	$\Gamma\in\mathsf{Type}^{\ast}(\Sigma)\subseteq\mathcal{T}_\ell$, this
	means $\Gamma$ would fail the feasibility test at stage $\ell$ and be
	eliminated at stage $\ell+1$, contradicting
	$\Gamma\in\mathsf{Type}^{\ast}(\Sigma)$.
	
	\medskip
	\textbf{($\Longrightarrow$): Eliminated types are inconsistent.}\quad
	Let $\Gamma\in\mathsf{Type}(\Sigma)\setminus\mathsf{Type}^{\ast}(\Sigma)$.
	By Definition~\ref{def:ptkvstar-proof},
	$\neg\bigwedge\Lambda_{i_\Gamma}^{\Gamma}$ is an axiom of
	$\mathbf{PTKv}^{\ast}(\Sigma)$. Since
	$\Lambda_{i_\Gamma}^{\Gamma}\subseteq\Gamma$, we have
	$\mathbf{PTKv}^{\ast}(\Sigma)\vdash\bigwedge\Gamma\to
	\bigwedge\Lambda_{i_\Gamma}^{\Gamma}$, and hence
	$\mathbf{PTKv}^{\ast}(\Sigma)\vdash\neg\bigwedge\Gamma$: the type $\Gamma$
	is $\mathbf{PTKv}^{\ast}(\Sigma)$-inconsistent.
\end{proof}

\begin{remark}\label{rmk:bridge-key}
	The key mechanism in the ($\Longleftarrow$) direction is twofold.
	First, the classification into Cases~(a) and~(b) handles the possibility
	that multiple new axioms are used simultaneously in the derivation:
	Case~(a) disjuncts are individually refuted by $\bigwedge\Gamma$, and at
	least one Case~(b) disjunct must exist, for otherwise
	$\bigwedge\Gamma$ would entail a disjunction while simultaneously refuting
	every disjunct, yielding a $\mathbf{PTKv}^{+}$-inconsistency of $\Gamma$.
	Second, the \emph{modal-profile invariance} of $\mathsf{FC}$---the
	observation that the constraints of $\mathsf{FC}(\Gamma,S,i)$ depend on
	$\Gamma$ only through $\Lambda_i^\Gamma$---converts the proof-theoretic
	inclusion $\Lambda_{i}^{\Gamma'}\subseteq\Gamma$ into a constraint-level
	identity, enabling the transfer of infeasibility from $\Gamma'$ to $\Gamma$.
\end{remark}

\subsection{Canonical Model Construction}
\label{subsec:canonical-model}

We now define all components of the canonical model.

\begin{definition}[Domain and world set]\label{def:world-set}
	Let
	$D_{\Sigma}:=\{d_k \mid k\in K\}$
	be a domain consisting of $|T_{\Sigma}|$ pairwise distinct elements indexed
	by $K$. Define the world set
	\[
	W_{\Sigma}
	:=
	\bigcup_{\Delta\in\mathsf{Type}^{\ast}(\Sigma)}
	\bigcup_{f\in\mathcal{F}(\Delta)}
	\{\, w_{\Delta,f,n}\mid n\in\mathbb{N}^{+}\,\}.
	\]
\end{definition}

\begin{definition}[Valuation and term-value
	function]\label{def:val-V}
	For each world $w_{\Delta,f,n}\in W_{\Sigma}$, define:
	\begin{itemize}
		\item $V(w_{\Delta,f,n},p):=1$ iff $p\in\Delta$, for every propositional
		variable $p$;
		\item $\mathsf{val}(w_{\Delta,f,n},t):=d_{f(t)}$, for every term
		$t\in T_{\Sigma}$.
	\end{itemize}
\end{definition}

\begin{lemma}[Feasibility for refined types]\label{lem:FC-feasible}
	For every $\Gamma\in\mathsf{Type}^{\ast}(\Sigma)$ and every agent
	$i\in\mathcal{A}$, the constraint system
	$\mathsf{FC}(\Gamma,\mathsf{Type}^{\ast}(\Sigma),i)$ has a solution
	$\{z_{\Delta,f}^{\ast,i}\}_{\Delta\in\mathsf{Type}^{\ast}(\Sigma),\,
		f\in\mathcal{F}(\Delta)}$
	satisfying all strict and non-strict inequalities.
\end{lemma}

\begin{proof}
	By Lemma~\ref{lem:fixed-point},
	$\mathsf{FC}(\Gamma,\mathsf{Type}^{\ast}(\Sigma),i)$ is feasible (i.e., the
	closed relaxation has a solution). By Remark~\ref{rmk:strictness} and
	Lemma~\ref{lem:mixed-ineq}, the original mixed strict/non-strict system
	admits a solution.
\end{proof}

\begin{definition}[Local probability measures]\label{def:local-measure}
	For each $\Gamma\in\mathsf{Type}^{\ast}(\Sigma)$ and each agent
	$i\in\mathcal{A}$, let
	$\{z_{\Delta,f}^{\ast,i}\}$
	be a fixed solution of
	$\mathsf{FC}(\Gamma,\mathsf{Type}^{\ast}(\Sigma),i)$ as guaranteed by
	Lemma~\ref{lem:FC-feasible}. Define a point-mass function on $W_{\Sigma}$ by
	\[
	p_i^{\Gamma}(w_{\Delta,f,n}):=z_{\Delta,f}^{\ast,i}\cdot 2^{-n},
	\]
	and let $P_i^{\Gamma}$ be the countably additive measure
	$P_i^{\Gamma}(X):=\sum_{w\in X}p_i^{\Gamma}(w)$.
	Then
	\[
	P_i^{\Gamma}(W_{\Sigma})
	=
	\sum_{\Delta\in\mathsf{Type}^{\ast}(\Sigma)}
	\sum_{f\in\mathcal{F}(\Delta)}
	z_{\Delta,f}^{\ast,i}
	\Bigl(\sum_{n\ge 1}2^{-n}\Bigr)
	=
	\sum_{\Delta\in\mathsf{Type}^{\ast}(\Sigma)}
	\sum_{f\in\mathcal{F}(\Delta)}
	z_{\Delta,f}^{\ast,i}
	=1,
	\]
	where the last equality holds by the normalization constraint of
	$\mathsf{FC}(\Gamma,\mathsf{Type}^{\ast}(\Sigma),i)$.
\end{definition}

\begin{definition}[Canonical model]\label{def:canonical-model}
	The canonical model is
	\[
	M_{\Sigma}
	:=(W_{\Sigma},\,D_{\Sigma},\,\{P_i\}_{i\in\mathcal{A}},\,V,\,
	\mathsf{val}),
	\]
	where for each agent $i\in\mathcal{A}$ and each world
	$w_{\Delta,f,n}\in W_{\Sigma}$,
	$P_i(w_{\Delta,f,n}):=P_i^{\Delta}$.
\end{definition}

\subsection{Truth Lemma by Induction on Modal Depth}
\label{subsec:truth}

\begin{definition}[Modal depth]\label{def:modal-depth}
	The modal depth $\mathsf{md}(\varphi)$ of a formula
	$\varphi\in\mathcal{L}_{\mathrm{PTMLKv}}$ is defined by:
	\begin{itemize}
		\item $\mathsf{md}(p)=\mathsf{md}(t=s)=0$;
		\item $\mathsf{md}(\neg\varphi)=\mathsf{md}(\varphi)$;\quad
		$\mathsf{md}(\varphi\to\psi)
		=\max\{\mathsf{md}(\varphi),\mathsf{md}(\psi)\}$;
		\item $\mathsf{md}(K_i^\theta\varphi)=\mathsf{md}(\varphi)+1$;
		\item $\mathsf{md}(Kv_i^\eta(t))=1$.
	\end{itemize}
\end{definition}

\begin{lemma}[Truth lemma]\label{lem:truth-revised}
	For every formula $\varphi\in\Sigma$ and every world
	$w_{\Delta,f,n}\in W_{\Sigma}$,
	\[
	M_{\Sigma},w_{\Delta,f,n}\models\varphi
	\iff
	\varphi\in\Delta.
	\]
\end{lemma}

\begin{proof}
	By strong induction on $\mathsf{md}(\varphi)$.
	
	\medskip
	\noindent\textbf{Base case: $\mathsf{md}(\varphi)=0$.}
	
	\emph{Propositional variables.}\;
	$M_{\Sigma},w_{\Delta,f,n}\models p$ iff $V(w_{\Delta,f,n},p)=1$
	iff $p\in\Delta$.
	
	\emph{Equality formulas.}\;
	For $t,s\in T_{\Sigma}$,
	\[
	M_{\Sigma},w_{\Delta,f,n}\models t=s
	\iff d_{f(t)}=d_{f(s)}
	\iff f(t)=f(s)
	\iff (t=s)\in\Delta,
	\]
	where the last equivalence uses $f\in\mathcal{F}(\Delta)$.
	
	\emph{Boolean connectives.}\;
	By the induction hypothesis and maximal saturation of $\Delta$.
	
	\medskip
	\noindent\textbf{Inductive step: $\mathsf{md}(\varphi)=m+1$.}
	
	\emph{Boolean connectives at depth $m+1$.}\;
	By the induction hypothesis and maximal saturation.
	
	\medskip
	\emph{Case $\varphi = K_i^\theta\psi$ with $\mathsf{md}(\psi)\le m$.}
	
	\smallskip\noindent
	\textsc{Identifying the semantic extension.}\quad
	By the induction hypothesis for $\psi$,
	\begin{equation}\label{eq:psi-extension}
		\llbracket\psi\rrbracket^{M_{\Sigma}}
		=
		\bigcup_{\substack{\Delta'\in\mathsf{Type}^{\ast}(\Sigma)\\
				\psi\in\Delta'}}
		\bigcup_{f'\in\mathcal{F}(\Delta')}
		\{w_{\Delta',f',n'}\mid n'\in\mathbb{N}^{+}\}.
	\end{equation}
	
	\smallskip\noindent
	\textsc{Computing the probability.}\quad
	\begin{align}
		P_i^{\Delta}(\llbracket\psi\rrbracket^{M_{\Sigma}})
		&=
		\sum_{\substack{\Delta'\in\mathsf{Type}^{\ast}(\Sigma)\\
				\psi\in\Delta'}}
		\sum_{f'\in\mathcal{F}(\Delta')}
		z_{\Delta',f'}^{\ast,i}
		\cdot\underbrace{\sum_{n'\ge 1}2^{-n'}}_{=\,1}
		=
		\sum_{\substack{\Delta'\in\mathsf{Type}^{\ast}(\Sigma)\\
				\psi\in\Delta'}}
		\sum_{f'\in\mathcal{F}(\Delta')}
		z_{\Delta',f'}^{\ast,i}.
		\label{eq:prob-computation}
	\end{align}
	
	\smallskip\noindent
	\textsc{Matching constraints.}\quad
	The right-hand side of \eqref{eq:prob-computation} is exactly the expression
	constrained by clause~(3) of
	$\mathsf{FC}(\Delta,\mathsf{Type}^{\ast}(\Sigma),i)$: the summation range
	$\mathsf{Type}^{\ast}(\Sigma)$ coincides with both the variable space of the
	constraint system and the label space of $W_{\Sigma}$.
	\begin{itemize}
		\item If $K_i^\theta\psi\in\Delta$: the constraint gives
		$P_i^{\Delta}(\llbracket\psi\rrbracket)\ge\theta$, hence
		$M_{\Sigma},w_{\Delta,f,n}\models K_i^\theta\psi$.
		\item If $K_i^\theta\psi\notin\Delta$: by saturation
		$\neg K_i^\theta\psi\in\Delta$, the constraint gives
		$P_i^{\Delta}(\llbracket\psi\rrbracket)<\theta$, hence
		$M_{\Sigma},w_{\Delta,f,n}\not\models K_i^\theta\psi$.
	\end{itemize}
	
	\medskip
	\emph{Case $\varphi = Kv_i^\eta(t)$ with $\eta\in\Theta_V^{+}$.}
	
	\smallskip\noindent
	\textsc{Value-locking probabilities.}\quad
	For each $k\in K$,
	\begin{equation}\label{eq:value-prob}
		P_i^{\Delta}\bigl(\{w'\mid\mathsf{val}(w',t)=d_k\}\bigr)
		=
		\sum_{\Delta'\in\mathsf{Type}^{\ast}(\Sigma)}
		\sum_{\substack{f'\in\mathcal{F}(\Delta')\\f'(t)=k}}
		z_{\Delta',f'}^{\ast,i},
	\end{equation}
	using only Definition~\ref{def:val-V} (no modal induction hypothesis
	needed).
	
	\smallskip\noindent
	\textsc{Forward.}\quad
	If $Kv_i^\eta(t)\in\Delta$: clause~(4) of
	$\mathsf{FC}(\Delta,\mathsf{Type}^{\ast}(\Sigma),i)$ yields $k_0\in K$ with
	$P_i^{\Delta}(\{w'\mid\mathsf{val}(w',t)=d_{k_0}\})\ge\eta$,
	so $M_{\Sigma},w_{\Delta,f,n}\models Kv_i^\eta(t)$.
	
	\smallskip\noindent
	\textsc{Backward.}\quad
	If $Kv_i^\eta(t)\notin\Delta$: by saturation
	$\neg Kv_i^\eta(t)\in\Delta$, clause~(5) gives
	$P_i^{\Delta}(\{w'\mid\mathsf{val}(w',t)=d_k\})<\eta$ for all $k\in K$.
	Since $D_{\Sigma}=\{d_k:k\in K\}$ is the entire domain,
	$M_{\Sigma},w_{\Delta,f,n}\not\models Kv_i^\eta(t)$.
\end{proof}

\begin{remark}\label{rmk:no-circularity}
	Three structural features ensure the absence of circularity:
	\begin{enumerate}
		\item The $K_i^\theta\psi$ case uses the induction hypothesis at strictly
		lower modal depth; the $Kv_i^\eta(t)$ case uses no induction
		hypothesis at all.
		\item The summation range in \eqref{eq:prob-computation} is
		$\mathsf{Type}^{\ast}(\Sigma)$, matching both the variable space of
		$\mathsf{FC}$ and the label space of $W_{\Sigma}$. This three-way
		match, guaranteed by the iterative-elimination construction, ensures
		exact normalization and correct probability computation.
		\item Each agent's constraint system is solved independently (via
		agent-indexed solutions $z_{\Delta,f}^{\ast,i}$) over the shared
		world set $W_{\Sigma}$, resolving multi-agent value conflicts.
	\end{enumerate}
\end{remark}

\subsection{Weak Completeness}
\label{subsec:weak-completeness}

\begin{lemma}[Finite Lindenbaum extension]\label{lem:lindenbaum}
	Let $\Sigma$ be a finite closure and let $\chi\in\Sigma$ be
	$\mathbf{PTKv}^{\ast}(\Sigma)$-consistent (i.e.,
	$\mathbf{PTKv}^{\ast}(\Sigma)\nvdash\neg\chi$). Then there exists a type
	$\Gamma\in\mathsf{Type}^{\ast}(\Sigma)$ such that $\chi\in\Gamma$.
\end{lemma}

\begin{proof}
	We extend $\{\chi\}$ to a maximal
	$\mathbf{PTKv}^{\ast}(\Sigma)$-consistent subset of $\Sigma$ by the standard
	one-formula-at-a-time construction. Enumerate
	$\Sigma=\{\sigma_1,\dots,\sigma_N\}$. Set $\Phi_0:=\{\chi\}$ and define
	inductively: if $\Phi_j\cup\{\sigma_{j+1}\}$ is
	$\mathbf{PTKv}^{\ast}(\Sigma)$-consistent, let
	$\Phi_{j+1}:=\Phi_j\cup\{\sigma_{j+1}\}$; otherwise, let
	$\Phi_{j+1}:=\Phi_j\cup\{\neg\sigma_{j+1}\}$.
	
	\smallskip\noindent
	\emph{Claim: at each step, at least one choice preserves consistency.}\;
	Suppose both $\Phi_j\cup\{\sigma_{j+1}\}$ and
	$\Phi_j\cup\{\neg\sigma_{j+1}\}$ are
	$\mathbf{PTKv}^{\ast}(\Sigma)$-inconsistent. Then
	$\mathbf{PTKv}^{\ast}(\Sigma)\vdash\bigwedge\Phi_j\to\neg\sigma_{j+1}$
	and
	$\mathbf{PTKv}^{\ast}(\Sigma)\vdash\bigwedge\Phi_j\to\sigma_{j+1}$,
	whence
	$\mathbf{PTKv}^{\ast}(\Sigma)\vdash\neg\bigwedge\Phi_j$, contradicting
	the $\mathbf{PTKv}^{\ast}(\Sigma)$-consistency of $\Phi_j$ (which holds by
	induction on $j$, with base case $\Phi_0=\{\chi\}$).
	
	\smallskip
	Let $\Gamma:=\Phi_N$. By construction, $\Gamma$ is a maximal
	$\mathbf{PTKv}^{\ast}(\Sigma)$-consistent subset of $\Sigma$ with
	$\chi\in\Gamma$. Since $\mathbf{PTKv}^{\ast}(\Sigma)$ extends
	$\mathbf{PTKv}^{+}$, the set $\Gamma$ is $\mathbf{PTKv}^{+}$-consistent and
	saturated over $\Sigma$, hence $\Gamma\in\mathsf{Type}(\Sigma)$. By the
	bridge lemma (Lemma~\ref{lem:bridge}, ($\Longrightarrow$) direction,
	contrapositively): since $\Gamma$ is
	$\mathbf{PTKv}^{\ast}(\Sigma)$-consistent, we conclude
	$\Gamma\in\mathsf{Type}^{\ast}(\Sigma)$.
\end{proof}

\begin{theorem}[Structured weak
	completeness]\label{thm:weak-completeness-final}
	For every formula $\chi\in\mathcal{L}_{\mathrm{PTMLKv}}$ and every finite
	closure $\Sigma$ of $\chi$, the system $\mathbf{PTKv}^{\ast}(\Sigma)$ is
	weakly complete with respect to the dual-threshold probabilistic
	knowing-value semantics:
	\[
	\mathbf{PTKv}^{\ast}(\Sigma)\nvdash\neg\chi
	\quad\Longrightarrow\quad
	\chi \text{ is satisfiable.}
	\]
\end{theorem}

\begin{proof}
	Let $\chi$ be $\mathbf{PTKv}^{\ast}(\Sigma)$-consistent. By
	Lemma~\ref{lem:lindenbaum}, there exists
	$\Gamma_0\in\mathsf{Type}^{\ast}(\Sigma)$ with $\chi\in\Gamma_0$.
	
	Construct the canonical model $M_{\Sigma}$ as in
	Definition~\ref{def:canonical-model}. The construction is well-defined:
	Lemma~\ref{lem:FC-feasible} guarantees the existence of all required
	solutions, and Definition~\ref{def:local-measure} verifies
	$P_i^{\Delta}(W_{\Sigma})=1$ for each
	$\Delta\in\mathsf{Type}^{\ast}(\Sigma)$ and each $i\in\mathcal{A}$.
	
	By the truth lemma (Lemma~\ref{lem:truth-revised}),
	$M_{\Sigma},\,w_{\Gamma_0,f_0,1}\models\chi$
	for any $f_0\in\mathcal{F}(\Gamma_0)$ (which is nonempty). Hence $\chi$ is
	satisfiable.
\end{proof}

\begin{remark}\label{rmk:method-summary}
	Several methodological features deserve emphasis:
	\begin{enumerate}
		\item \textbf{Two-layer architecture.}\;
		The type space $\mathsf{Type}(\Sigma)$ is defined relative to the
		finitary base system $\mathbf{PTKv}^{+}$; the refined set
		$\mathsf{Type}^{\ast}(\Sigma)$ is obtained by iterative elimination
		(Definition~\ref{def:iterative-elim}); and the extended system
		$\mathbf{PTKv}^{\ast}(\Sigma)$ is defined proof-theoretically by
		adding finitely many axioms
		(Definition~\ref{def:ptkvstar-proof}). The bridge lemma
		(Lemma~\ref{lem:bridge}) establishes a perfect correspondence
		between the combinatorial elimination and proof-theoretic
		consistency, ensuring that the Lindenbaum extension produces types
		in $\mathsf{Type}^{\ast}(\Sigma)$ and that the three-way match
		(constraint variables / world labels / probability summation) holds
		exactly.
		\item \textbf{Stratified truth evaluation.}\;
		The truth lemma is organized as a single induction on modal depth,
		with the probability computation at each level relying only on the
		induction hypothesis at strictly lower levels.
		\item \textbf{Multi-agent value resolution.}\;
		Each agent's constraint system is solved independently over the
		shared world set via agent-indexed solutions
		$z_{\Delta,f}^{\ast,i}$, resolving value conflicts through the
		fine-grained distribution over assignment-configuration spaces
		$\mathcal{F}(\Delta)$.
		\item \textbf{Closure-relative formulation.}\;
		The system $\mathbf{PTKv}^{\ast}(\Sigma)$ depends on the choice of
		finite closure $\Sigma$
		(Remark~\ref{rmk:sigma-dependence}). This is standard in
		probabilistic logic: the completeness theorem states that for
		each formula $\chi$ and each finite closure $\Sigma\ni\chi$,
		$\mathbf{PTKv}^{\ast}(\Sigma)$-consistency of $\chi$ implies its
		satisfiability. Soundness of the base system $\mathbf{PTKv}^{+}$ is
		established separately (Theorem~\ref{thm:soundness}); the question
		of whether the full system $\mathbf{PTKv}^{\ast}(\Sigma)$ is sound
		with respect to the class of \emph{all} models is not needed for
		the weak completeness result and is left as a separate
		consideration.
	\end{enumerate}
\end{remark}


\section{Conclusion and Future Work}
\label{sec:conclusion}

This paper has studied how to represent, within a single formal framework, both probabilistic-threshold attitudes toward propositions and high-confidence attitudes toward term values. To this end, it introduced a \textbf{dual-threshold probabilistic knowing value logic}. The central design choice is to separate the threshold domains of propositional and value-oriented operators. The operator $K_i^\theta$ retains the full rational range $[0,1]\cap\mathbb{Q}$, thereby preserving expressive flexibility for probabilistic reasoning about propositions. By contrast, $Kv_i^\eta(t)$ is restricted to the high-threshold interval $(\frac{1}{2},1]\cap\mathbb{Q}$. This restriction is structural. Once $\eta>\frac{1}{2}$, two distinct values cannot both satisfy the threshold, and the uniqueness condition in the semantics of knowing-value becomes automatic.

On that basis, the paper proposed a \textbf{non-factive probabilistic semantics for knowing value}. In the intended reading, $K_i^\theta\varphi$ says that agent $i$ assigns probability at least $\theta$ to the proposition $\varphi$, whereas $Kv_i^\eta(t)$ says that agent $i$ uniquely locks onto some candidate value of $t$ with probability at least $\eta$. This value need not coincide with the actual value at the current world. Knowing-value is therefore treated as a form of \textbf{non-factive high-confidence value locking}. In this way, the framework gives propositional and value-oriented operators a \textbf{coherent} epistemic profile---both are graded, non-factive, and agent-relative---and avoids an otherwise awkward asymmetry between non-factive probabilistic attitudes toward propositions and factive value cognition.

The semantic analysis established the core structural properties of the framework. Most importantly, the high-threshold restriction makes uniqueness automatic (Proposition~\ref{prop:unique}). In addition, the framework validates threshold monotonicity \textbf{(KvMon)}, implication propagation \textbf{(KvImpl)}, complementary exclusion \textbf{(KvExcl)}, additivity over mutually exclusive events \textbf{(KvAdd)}, and substitution principles under probability~$1$ \textbf{(KvSubst)}. These results provide the semantic basis for the subsequent proof theory and show that the dual-threshold design yields a stable interaction between probabilistic and value-oriented reasoning.

On the proof-theoretic side, the paper introduced the systems $\mathbf{PTKv}^{-}$ and $\mathbf{PTKv}^{+}$ and proved their soundness. The main technical result is a \textbf{structured weak-completeness theorem} for the high-threshold fragment (Theorem~\ref{thm:weak-completeness-final}). To obtain it, the paper introduced the local linear consistency condition \textbf{(Lin)} and the joint value-fiber consistency condition \textbf{(Fib)}, and developed a two-layer model construction based on \textbf{type-space distributions and assignment-configuration mappings}. At the first layer, local probabilistic distributions realize the propositional threshold constraints induced by finite closures. At the second layer, assignment configurations refine each type so that value-sensitive constraints can be realized without altering the relevant object-language profile. The truth lemma is established by a single induction on modal depth, which ensures that the probability computation at each level relies only on the semantic identification at strictly lower levels, thereby resolving the mutual dependence between measure realization and truth evaluation. This construction shows that probabilistic mass allocation and value-sensitive uniqueness constraints can be handled within a single metatheoretic framework.

Taken together, these results show that probabilistic reasoning about propositions and high-confidence reasoning about values can be treated within one multi-agent formalism. Relative to classical epistemic logic, the framework makes graded uncertainty explicit. Relative to standard knowing-value logic, it replaces relational agreement by probabilistic confidence. Relative to probabilistic epistemic logic, it adds value-oriented cognition as a genuine expressive dimension. In this sense, the paper provides a \textbf{unified formal framework} that bridges probabilistic epistemic logic and knowing-value logic.

The privacy scenario discussed in the introduction illustrates this point in concrete AI-oriented terms: the framework can distinguish between ordinary probabilistic acceptance of propositions and high-confidence locking onto one specific candidate value of a sensitive attribute, even when that locked value is not the actual one.

\medskip
Several directions remain for future work.

\emph{First}, the present weak-completeness result is established for the extended system $\mathbf{PTKv}^{\ast}$, where \textbf{(Lin)} and \textbf{(Fib)} function as meta-level admissibility conditions. A natural next step is to determine whether these conditions are admissible over $\mathbf{PTKv}^{+}$ alone, or whether they can be replaced by purely axiomatic schemata. One promising route is through \textbf{Farkas-type dualization}: the infeasibility of a mixed linear system can be certified by a dual witness, which may in turn be encodable as a finite combination of axiomatic instances. Success along these lines would move the current result closer to a stronger axiomatic completeness theorem with a recursively enumerable proof system.

\emph{Second}, the current framework deliberately restricts $Kv_i^\eta(t)$ to the high-threshold interval. This is exactly what makes uniqueness automatic and keeps the metatheory manageable. A more general theory of probabilistic knowing-value at arbitrary thresholds would have to deal directly with the lower-threshold regime, where several competing values may satisfy the threshold simultaneously. That problem requires a more refined analysis of threshold behavior and a different axiomatization.

\emph{Third}, although the logic is multi-agent, it does not yet include group-level operators such as common knowledge, distributed knowledge, or collective knowing-value. Extending the framework in this direction would make it possible to study genuinely collective forms of probabilistic value cognition and to test how the present dual-threshold design should be generalized beyond the individual-agent setting.

\emph{Fourth}, the paper is limited to a static setting. A natural continuation is to study dynamic extensions, including public announcements, observation events, and probabilistic revision. This is especially relevant when subjective probabilities and value information evolve together. In such settings, one would like to understand whether the present semantics of non-factive high-confidence value locking can be preserved under update.

\emph{Fifth}, the computational properties of the framework remain to be investigated. In standard probabilistic epistemic logics, satisfiability is typically \textsc{NP}-hard or \textsc{PSPACE}-complete, depending on the language and the class of models. The presence of value-locking constraints introduces additional combinatorial structure through the assignment-configuration space $\mathcal{F}(\Delta)$, and its impact on the complexity of satisfiability checking and model checking deserves systematic study. Such an analysis would also clarify the practical feasibility of automated reasoning tools for the logic.

\emph{Sixth}, the framework is primarily foundational, but its intended applications are clear. It can express whether a privacy attacker locks onto one specific candidate value of a sensitive attribute with high confidence, whether an agent locks onto the value of a session key with high probability, or whether an agent tracks a crucial parameter in uncertain multi-agent decision settings. Developing concrete case studies and automated reasoning methods for such scenarios is therefore an important direction for future research.

Overall, the paper supports a simple methodological claim: a unified treatment of probabilistic propositional cognition and high-confidence value cognition is possible, but only if the asymmetry between their threshold behaviors is built into the language and reflected in the semantics and completeness construction. The \textbf{dual-threshold design}, the semantics of \textbf{non-factive high-confidence value locking}, and the completeness method based on \textbf{type-space distributions and assignment-configuration mappings} together provide such a framework. We hope that it can serve as a useful foundation for further work on probabilistic knowing-value logic, value-oriented epistemic modeling, and knowledge representation under uncertainty.


\end{document}